\documentclass{eptcs-modified}

\usepackage{amsmath}
\usepackage{amssymb}
\usepackage{url}
\usepackage{amsthm}
\usepackage{graphicx}
\usepackage{amscd}
\usepackage[usenames,dvipsnames,svgnames,table]{xcolor}

\usepackage{tikz}
\usetikzlibrary{positioning}
\usetikzlibrary{shapes.geometric}

\begin{document}

\theoremstyle{definition}
\newtheorem{theorem}{Theorem}
\newtheorem{definition}[theorem]{Definition}
\newtheorem{problem}[theorem]{Problem}
\newtheorem{assumption}[theorem]{Assumption}
\newtheorem{corollary}[theorem]{Corollary}
\newtheorem{proposition}[theorem]{Proposition}
\newtheorem{example}[theorem]{Example}
\newtheorem{lemma}[theorem]{Lemma}
\newtheorem{observation}[theorem]{Observation}
\newtheorem{fact}[theorem]{Fact}
\newtheorem{question}[theorem]{Open Question}
\newtheorem{conjecture}[theorem]{Conjecture}
\newtheorem{addendum}[theorem]{Addendum}
\newcommand{\uint}{{[0, 1]}}
\newcommand{\Cantor}{{\{0,1\}^\mathbb{N}}}
\newcommand{\name}[1]{\textsc{#1}}
\newcommand{\me}{\name{P.}}
\newcommand{\id}{\textrm{id}}
\newcommand{\dom}{\operatorname{dom}}
\newcommand{\Dom}{\operatorname{Dom}}
\newcommand{\codom}{\operatorname{CDom}}
\newcommand{\Baire}{{\mathbb{N}^\mathbb{N}}}
\newcommand{\hide}[1]{}
\newcommand{\mto}{\rightrightarrows}
\newcommand{\Sierp}{Sierpi\'nski }
\newcommand{\BC}{\mathcal{B}}
\newcommand{\C}{\textrm{C}}
\newcommand{\lpo}{\textrm{LPO}}
\newcommand{\leqW}{\leq_{\textrm{W}}}
\newcommand{\leW}{<_{\textrm{W}}}
\newcommand{\equivW}{\equiv_{\textrm{W}}}
\newcommand{\equivT}{\equiv_{\textrm{T}}}
\newcommand{\geqW}{\geq_{\textrm{W}}}
\newcommand{\pipeW}{|_{\textrm{W}}}
\newcommand{\nleqW}{\nleq_\textrm{W}}

\newcommand{\IQ}{\mathcal{I}(\mathbb{Q}[X_1,\dots,X_d])}
\newcommand{\QP}{\mathbb{Q}[X_1,\dots,X_d]}
\newcommand{\height}{\operatorname{ht}}

\newcommand{\Sort}{\textsf{Sort}}
\newcommand{\isInfinite}{\textsf{isInfinite}}
\newcommand{\isDefined}{\textsf{isDefined}}
\newcommand{\isFinite}{\textsf{isFinite}}

\title{A topological view on algebraic computation models}

\author{
Eike Neumann
\institute{Aston University, Birmingham, UK}
\email{neumaef1@aston.ac.uk}
\and
Arno Pauly
\institute{Computer Laboratory, University of Cambridge, UK\thanks{Pauly has since moved to the D\'epartement d'Informatique, Universit\'e libre de Bruxelles, Belgium.} \\ \& \\ University of Birmingham, Birmingham, UK}
\email{Arno.Pauly@cl.cam.ac.uk}
}

\def\titlerunning{A topological view on algebraic computation models}
\def\authorrunning{E. Neumann \& A. Pauly}
\maketitle

\begin{abstract}
We investigate the topological aspects of some algebraic computation models, in particular the BSS-model. Our results can be seen as bounds on how different BSS-computability and computability in the sense of computable analysis can be. The framework for this is Weihrauch reducibility. As a consequence of our characterizations, we establish that the solvability complexity index is (mostly) independent of the computational model, and that there thus is common ground in the study of non-computability between the BSS and TTE setting.
\end{abstract}

\section{Introduction}
There are two major paradigms for computability on functions on the real numbers: On the one hand, computable analysis in the tradition of \name{Grzegorczyk} \cite{grzegorczyk,grzegorczyk2} and \name{Lacombe} \cite{lacombe} as championed by \name{Weihrauch} \cite{weihrauch,weihrauchd} (see also the equivalent approaches by \name{Pour-El} and \name{Richards} \cite{pourel} or \name{Ko} \cite{ko}). On the other hand, the BSS-machines by \name{Blum}, \name{Shub} and \name{Smale} \cite{blum,blum2}, or the very similar real-RAM model. Incidentally, both schools claim to be in the tradition of \name{Turing}.

Computable analysis can, to a large extent, be understood as effective topology \cite{escardo,pauly-synthetic} -- this becomes particularly clear when one moves beyond just the real numbers, and is interested in computability on spaces of subsets or functionals. In particular, we find that the effective Borel hierarchy occupies the position analogous to the arithmetical hierarchy in classical recursion theory; and that incomputability of natural problems is typically a consequence of discontinuity. A more fine-grained view becomes possible in the framework of Weihrauch reducibility (more below).

In contrast, the study of BSS-computability is essentially a question akin to (logical) definability in algebraic structures. This causes the lack of a stable notion of BSS-computability on \emph{the reals}: BSS-computability on the ring $(\mathbb{R},+,\times,=)$ differs from BSS-computability on the unordered field $(\mathbb{R},+,\times,-,/,=)$, which in turn differs from BSS-computability on the ordered field $(\mathbb{R},+,\times,-,/,<)$. Taking into account as basic functions further maps such as square root or the exponential function induces additional variants. There are, however, also topological obstacles to being BSS-computable -- and as we shall demonstrate, these obstacles are common to all variants of BSS-machines.

\name{Hotz} and his coauthors \cite{hotz2,hotz3,gaertnerhotz} have introduced and studied a number of extensions of BSS-machines (called \emph{analytic machines}), which can (in various ways) make use of approximations. These additional features enable the ``computation'' of even more discontinuous functions, however, there are still topological limitations.

In the present article, we characterize the maximal degrees of discontinuity for functions computable by each of the algebraic machine models in the framework of Weihrauch reducibility. We thus continue and extend the work by \name{G\"artner} and \name{Ziegler} in \cite{zieglergaertner}. This allows us to differentiate between topological and algebraic reasons for non-computability in these models. Moreover, we can show that for sufficiently discontinuous problems the difference between the various computational models vanishes: This can be formalized using \name{Hansen}'s \emph{solvability complexity index} (SCI) \cite{hansen,hansen2}, which is the least number of limits one needs in order to solve a particular problem over some basic computational model. Our results show e.g.~that for values of 2 or greater the SCI based on the BSS-model coincides with the SCI based on the computable analysis model.

In Section 3 we characterise the computational power of generalized register machines with $\lpo$-tests (Proposition \ref{Proposition: C_N equiv_W lpo-diamond}) as well as BSS-machines over the reals with equality- and order-tests (Corollary \ref{corr:bss}). We show that their strength is captured by the Weihrauch degree $\C_\mathbb{N}$ of closed choice over the natural numbers, in the sense that any function computable by either type of machine is Weihrauch-reducible to $\C_\mathbb{N}$ and for each type of machine there exists a function which is Weihrauch-equivalent to $\C_\mathbb{N}$ and computable by that type of machine. We furthermore show that all BSS-computable \emph{total} functions are strictly below $\C_\mathbb{N}$ (Corollary \ref{Corollary: total BSS upper bound} and Proposition \ref{Proposition: total BSS lower bound}) and that the strength of BSS-machines with equality-tests but without order-tests is characterised by the strictly weaker Weihrauch degree $\lpo^*$ (Proposition \ref{Proposition: BSS with equality}).

In Section 4 we study the BSS-halting problem from the point of view of Weihrauch-reducibility. We show in particular that the halting problem of a machine over the signature $(\mathbb{R}, +, =, <)$ cannot be solved by any machine over a signature which expands $(\mathbb{R}, +, =, <)$ by continuous operations only. Therefore, the inability of BSS-machines to solve their halting problem already holds for topological reasons.

In Section 5 we prepare the discussion of analytic machines and the solvability complexity index by some more technical results on Weihrauch degrees. We introduce the Weihrauch degree of the problem $\Sort$ of ``sorting'' an infinite binary sequence. We discuss its position in the Weihrauch lattice in detail and prove in particular that $\Sort^n <_W \Sort^{n + 1}$ (Corollary \ref{Corollary: Sort^n < Sort^(n + 1)}) and that we have the absorption $\lim \star \lim \star \Sort \equivW \lim \star \lim$ (Corollary \ref{corr:low2}).

In Section 6 we characterise the strength of analytic machines. We show that the computational power of analytic machines is characterised by the Weihrauch degree $\Sort^*$ in the same way as the power of BSS-machines is characterised by $\C_\mathbb{N}$ (Observation \ref{obs:analytic} and Corollary \ref{corr:algdacsortstar}).

In Section 7 we combine the results of the two previous sections to show that the SCI of an uncomputable function over the BSS-model is the same as the SCI over the computable analysis model as soon as the SCI over either model is greater than or equal to 2 (Theorem \ref{theo:sci}).

\section{Background on the models and Weihrauch reducibility}

\subsection{Algebraic computation models}
We introduce the notion of a register machine over some algebraic structure, following \name{Ga\ss{}ner} \cite{gassner4,gassner3,gassner2} (1997+) and \name{Tavana} and \name{Weihrauch} (2011) \cite{tavana}. Other approaches to computation over algebraic structures were put forth e.g.~by \name{Tucker} and \name{Zucker} (2000) \cite{zucker} and \name{Hemmerling} (1998) \cite{hemmerling}. For this consider some algebraic structure $\mathfrak{A} = (A, f_1, f_2, \ldots, T_1, T_2, \ldots)$, where $A$ is a set, each $f_i$ is a (partial) function of type $f_i : \subseteq A^{k_i} \to A$, and each $T_i$ is a relation of type $T_i \subseteq A^{l_i}$. In the usual examples, the signatures will be finite, but this is not essential for our considerations.

Generalized register machines will compute functions of type $g : A^* \to A^*$. They have registers $(R_i)_{i \in \mathbb{N}}$ holding elements of $A$, and index registers $(I_n)_{n \in \mathbb{N}}$ holding natural numbers. Programs are finite lists of commands, consisting of:

\begin{itemize}
 \item standard register machine operations on the index registers
 \item copying the value of the register $R_{I_1}$ indexed by $I_1$ into $R_{I_0}$
 \item applying some $f_i$ to the values contained in $R_1, \ldots, R_{k_i}$ and writing the result into $R_0$
 \item branching to a line in the program depending on the value of $T_i$ on the values contained in $R_1, \ldots, R_{l_i}$
 \item \textrm{HALT}, in which case the values currently in the registers $R_0, \ldots, R_{I_0}$ constitute the output
\end{itemize}

Initially, the register $I_0$ contains the length of the input, all other $I_n$ start at $0$. The input is in $R_1, \ldots, R_n$, all other $R_i$ contain some fixed value $a_0 \in A$. If the program either fails to halt on some input, or invokes a partial function on some values outside its domain, the computed function is undefined on these values. We call a (partial) function $g :\subseteq A^* \to A^*$ $\mathfrak{A}$-computable, if there is a generalized register machine program computing it.

In analogy to the situation in computable analysis, we shall call sections of $\mathfrak{A}$-computable function $\mathfrak{A}$-continuous. This boils down to programs being allowed additional assignment operations $R_i := a$ for any element $a \in A$. These are usually included in the algebraic models, but have the undesirable consequence of there being more than countably many programs, if $A$ itself is uncountable.

The primary example of a structure will be $(\mathbb{R}, +, \times, <)$ (yielding BSS-computability). Secondary examples include $(\mathbb{R}, +, =)$ (additive machines with equality) and other combinations of continuous functions and tests in $\{=, <\}$.

\subsection{Analytic machines}
The analytic machines introduced in \cite{hotz2,hotz3} enhance BSS-machines by means to approximate functions. More generally, if we consider the algebraic structure $\mathfrak{A}$ to also carry a metric, we can define functions computable by strongly $\mathfrak{A}$-analytic machines and functions computable by weakly $\mathfrak{A}$-analytic machines in one of two equivalent ways: Either a generalized register machine receives an additional input $n \in \mathbb{N}$; thus computes a function $G : \subseteq \mathbb{N} \times A^* \to A^*$, which is then considered as weakly approximating $g : \subseteq A^* \to A^*$ iff $\forall \mathbf{a} \in \dom(g) \ \lim_{n \to \infty} G(n,\mathbf{a}) = g(a)$, and as strongly approximating $g$ iff $\forall \mathbf{a} \in \dom(g) \ \forall n \in \mathbb{N} \ d(G(n,\mathbf{a}), g(a)) < 2^{-n}$. Alternatively there is no extra input, and the machine keeps running, and produces an infinite sequence $p \in A^\omega$ (plus some information on the length of the desired output). The limit conditions are the same.

We shall speak just of ``functions computable by a strongly (weakly) analytic machine'' if the underlying structure is $(\mathbb{R}, +, \times, <)$. As functions computable by an analytic machine receive their input exactly, but produce their output in an approximative fashion, they are generally not closed under composition.

Further variants of analytic machines have been considered, which are not relevant for the present paper though. We refer to \cite{ziegler4} by \name{Ziegler} for an excellent discussion.

\subsection{Type-2 Turing machines}
The fundamental model for computable analysis/the TTE-framework \cite{weihrauchd} are the Type-2 Turing machines. Structurally, these do not differ from ordinary Turing machines with a designated write-once only output tape. What differs is that instead of halting and thus producing a finite output, the machine continues to run for ever. As long as it keeps writing on the output tape, this produces an infinite sequence in the limit. While some might object to the use of infinitely long computations, this model is realistic in as far as anything written on the output tape at some finite time constitutes a prefix to the infinite output (as each cell in the output tape can be changed just once). Thus, this model inherently captures approximating computations -- and its intricate connection to topology is perhaps unsurprising. In particular, we find that any computable function is automatically continuous w.r.t.~the standard topology on $\Cantor$ -- and vice versa, every continuous function becomes computable relative to some oracle.

Type-2 machines natively provide us with a notion of computability on $\Cantor$. This is then transferred to the spaces of actual interest by means of a representation. A \emph{represented space} is a pair $\mathbf{X} = (X, \delta_X)$ of a set $X$ and a partial surjection $\delta_X : \subseteq \Cantor \to X$. A (multi-valued) function between represented spaces is a partial relation $f \subseteq X \times Y$. We write $f : \subseteq \mathbf{X} \mto \mathbf{Y}$ for these; here $\subseteq$ denotes (potential) partiality, and $\mto$ (potential) multi-valuedness. We write $f(x) = \{y \in Y \mid (x,y) \in f\}$ and $\dom(f) := \{x \in \mathbf{X} \mid \exists y \in f(x)\}$. We recall that composition of multi-valued functions is defined via $\dom(g \circ f) := \{x \in \dom(f) \mid f(x) \subseteq \dom(g)\}$ and $z \in (g \circ f)(x)$ for $x \in \dom(g \circ f)$, if $\exists y \in f(x) \ \ z \in g(y)$.

For $f : \subseteq \mathbf{X} \mto \mathbf{Y}$ and $F : \subseteq \Cantor \to \Cantor$, we call $F$ a realizer of $f$ (notation $F \vdash f$), iff $\delta_Y(F(p)) \in f(\delta_X(p))$ for all $p \in \dom(f\delta_X)$. A function between represented spaces is called computable (continuous), iff it has a computable (continuous) realizer. Note that, unlike in the case of algebraic models, the behaviour of a Type-2 machine which computes a partial function is completely unconstrained on inputs outside of the function's domain.

\begin{figure}[h]
 $$\begin{CD}
\Cantor @>F>> \Cantor\\
@VV\delta_\mathbf{X}V @VV\delta_\mathbf{Y}V\\
\mathbf{X} @>f>> \mathbf{Y}
\end{CD}$$
\caption{The notion of a realizer}
\end{figure}

As we are primarily interested in computability on $\mathbb{R}$ and $\mathbb{R}^*$, we shall introduce the standard representations for these spaces. Fix some standard enumeration $\nu_\mathbb{Q} : \mathbb{N} \to \mathbb{Q}$. Then let $\rho(p) = x$ if $p = 0^{n_0}10^{n_1}1\ldots$ and $\forall i \in \mathbb{N} \ \ d(x,\nu_\mathbb{Q}(n_i)) < 2^{-i}$. In other words, a name for a real number is a sequence of rationals converging to it with some prescribed speed. We thus understand $\mathbb{R}$ to be the represented space $(\mathbb{R}, \rho)$. As we can form products and coproducts of represented spaces, we automatically obtain a representation for $\mathbb{R}^* = \coprod_{n \in \mathbb{N}} \mathbb{R}^n$.

We will encounter decision problems, and thus need spaces of truth-values. For this, we use both the space $\{0,1\}$ represented by $\delta_{\mathbf{2}}$ defined by $\delta_\mathbf{2}(p) = p(0)$, as well as Sierpi\'nski space $\mathbb{S}$. The latter has the underlying set $\{\bot,\top\}$ and the representation $\delta_\mathbb{S}$ with $\delta_\mathbb{S}(0^\omega) = \bot$ and $\delta_\mathbb{S}(p) = \top$ iff $p \neq 0^\omega$. As usual, we identify $0$ with $\bot$ and $1$ with $\top$. The space $\{0,1\}$ captures decidability, and $\mathbb{S}$ captures semi-decidability. The usual boolean connectives $\wedge$ and $\vee$ are computable on both spaces. Negation $\neg : \{0,1\} \to \{0,1\}$ is computable, whereas $\neg : \mathbb{S} \to \mathbb{S}$ is not computable.

We also make use of the represented space $\mathbb{N}$, represented via $\delta_\mathbb{N}^{-1}(n) = \{0^n1^\omega\}$. Any represented space naturally comes with a topology, namely the final topology along the representation, where the domain of the representation just inherits the subspace topology of the usual complete metric on $\Cantor$. For $\mathbb{N}$, this is the discrete topology, for $\mathbb{R}$ the usual Euclidean topology. Via the utm-theorem, we obtain the represented space $\mathcal{O}(\mathbf{X})$ of the open subsets of the represented space $\mathbf{X}$ in a canonical manner, by identifying them with continuous functions from $X$ to $\mathbb{S}$. By considering complements instead we obtain a representation of the space $\mathcal{A}(\mathbf{X})$. Here, the space $\mathcal{A}(\mathbf{X})$ is the space of closed subsets represented with \emph{negative} information, or equipped with the upper Fell topology. A representation via \emph{positive} information is obtained by identifying a closed set $A \subseteq \mathbf{X}$ with the open set of all open sets which intersect $A$. The corresponding represented space is denoted by $\mathcal{V}(\mathbf{X})$ and called the space of \emph{overt subsets} of $\mathbf{X}$. We only use the special cases $\mathcal{O}(\mathbb{N})$ and $\mathcal{A}(\mathbb{N})$ here though. One may consider $\mathcal{O}(\mathbb{N})$ to be represented by $\delta\colon \Cantor \to \mathcal{O}(\mathbb{N})$ with $n \in \delta(p)$ iff $01^{n+1}0$ is a subword of $p$; and thus $\mathcal{A}(\mathbb{N})$ by $\psi\colon \Cantor \to \mathcal{A}(\mathbb{N})$ with $n \in \psi(p)$ iff $01^{n+1}0$ is not a subword of $p$. Note that the computable points of $\mathcal{O}(\mathbb{N})$ are precisely the computably enumerable sets, and the computable points of $\mathcal{A}(\mathbb{N})$ are the co-c.e. sets.

\subsection{Weihrauch reducibility}
Weihrauch reducibility is a preorder on multivalued functions between represented spaces, and serves as a framework for comparing incomputability in the Type-2 setting, similar to the role of many-one or Turing reductions in classical recursion theory:

\begin{definition}[Weihrauch reducibility]
\label{def:weihrauch}
Let $f,g$ be multi-valued functions on represented spaces.
Then $f$ is said to be {\em Weihrauch reducible} to $g$, in symbols $f\leqW g$, if there are computable
functions $K,H:\subseteq\Cantor\to\Cantor$ such that $K\langle \id, GH \rangle \vdash f$ for all $G \vdash g$.
Accordingly, $f$ is said to be {\em continuously Weihrauch reducible} to $g$, in symbols $f \leqW^c g$, if there exist continuous functions $K$ and $H$ satisfying this condition.
\end{definition}

The relation $\leqW$ is reflexive and transitive. We use $\equivW$ to denote equivalence regarding $\leqW$,
and by $\leW$ we denote strict reducibility. By $\mathfrak{W}$ we refer to the partially ordered set of equivalence classes. As shown in \cite{paulyreducibilitylattice,brattka2}, $\mathfrak{W}$ is a distributive lattice, and also the usual product operation on multivalued function induces an operation $\times$ on $\mathfrak{W}$. The algebraic structure on $\mathfrak{W}$ has been investigated in further detail in \cite{paulykojiro,paulybrattka4}.

There are two relevant unary operations defined on $\mathfrak{W}$, both happen to be closure operators. The operation $^*$ was introduced in \cite{paulyreducibilitylattice,paulyincomputabilitynashequilibria} by setting $f^0 := \id_\Baire$, $f^{n+1} := f \times f^{n}$ and then $f^*(n,x) := f^n(x)$. It corresponds to making any finite number of parallel uses of $f$ available. Similarly, the \emph{parallelization} operation $\widehat{\phantom{f}}$ from \cite{brattka2,brattka3} makes countably many parallel uses available by $\widehat{f}(x_0, x_1, x_2, \ldots) := (f(x_0), f(x_1), f(x_2), \ldots)$.

We will make use of an operation $\star$ defined on $\mathfrak{W}$ that captures aspects of function composition. Following \cite{gherardi4,paulybrattka3cie}, let $f \star g := \max_{\leqW} \{f_0 \circ g_0 \mid f \equivW f_0 \wedge g \equivW g_0\}$. We understand that the quantification is running over all suitable functions $f_0$, $g_0$ with matching types for the function composition. It is not obvious that this maximum always exists, this is shown in \cite{paulybrattka4} using an explicit construction for $f \star g$. Like function composition, $\star$ is associative but generally not commutative. We use $\star$ to introduce iterated composition by setting $f^{(0)} := \id_\Baire$ and $f^{(n+1)} = f^{(n)} \star f$.

All computable multivalued functions with a computable point in their domain are Weihrauch equivalent, this degree is denoted by $1$.

An important source for examples of Weihrauch degrees that are relevant for the classification of theorems are the closed choice principles studied in e.g.~\cite{brattka3,paulybrattka}:
\begin{definition}
Given a represented space $\mathbf{X}$, the associated closed choice principle $\C_\mathbf{X}$ is the partial multivalued function $\C_\mathbf{X} : \subseteq \mathcal{A}(\mathbf{X}) \mto \mathbf{X}$ mapping a non-empty closed set to an arbitrary point in it.
\end{definition}

The Weihrauch degree corresponding to $\C_{\mathbb{N}}$ has received significant attention, e.g.~in \cite{brattka3,paulybrattka,paulymaster,mylatz,mylatzb,hoelzl,paulyoracletypetwo,pauly-fouche2,pauly-steinberg}.
			In particular, as shown in \cite{paulydebrecht}, a function between computable Polish spaces is Weihrauch reducible to $\C_\mathbb{N}$ iff it is piecewise computable iff it is effectively $\Delta^0_2$-measurable.

The second standard Weihrauch degree very relevant for our investigation will be $\lim$, with its representative $\lim : \subseteq \Baire \to \Baire$ defined via $\lim(p)(n) = \lim_{i \to \infty} p(\langle n, i\rangle)$. It was shown in \cite{brattka} that $\lim$ is Weihrauch-complete for $\Sigma^0_2$-measurable functions, and that, more generally, $\lim^{(n)}$ is Weihrauch-complete for $\Sigma^0_{n+1}$-measurable function. This line of research was continued in \cite{nobrega}.

The third standard Weihrauch degree we will refer to is $\lpo$, which has the eponymous representative $\lpo : \Cantor \to \{0,1\}$ mapping $0^\omega$ to $1$ and $p \neq 0^\omega$ to $0$. By virtue of having the same realizers, we also find $\id : \mathbb{S} \to \{0,1\}$ in this class. Furthermore, $\mathalpha{=0} : \mathbb{R} \to \{0,1\}$, $\mathalpha{=} : \mathbb{R} \times \mathbb{R} \to \{0,1\}$ and $\mathalpha{<} : \mathbb{R} \times \mathbb{R} \to \{0,1\}$ are members of the Weihrauch degree $\lpo$. In the context of computable analysis the degree was first introduced and named as such in \cite{weihrauchc}, based on earlier usage in constructive mathematics \cite{bridges6}.

For our purposes, the following representatives and properties of the degree $\C_\mathbb{N}$ are also relevant:

		\begin{lemma}[\cite{pauly-fouche2}]\label{lemma:cn}
			The following are Weihrauch equivalent:
			\begin{enumerate}
				\item $\C_\mathbb{N}$, that is closed choice on the natural numbers.
				\item $\textrm{UC}_\mathbb{N}$, defined via $\textrm{UC}_\mathbb{N} = \left ( \C_\mathbb{N} \right)|_{\{A \in \mathcal{A}(\mathbb{N}) \mid |A| = 1\}}$.
				\item $\min : \subseteq \mathcal{A}(\mathbb{N}) \to \mathbb{N}$ defined on the non-empty closed subsets of $\mathbb{N}$.
				\item $\max_{\mathcal{O}} : \subseteq \mathcal{O}(\mathbb{N}) \to \mathbb{N}$ defined on the non-empty bounded open subsets of $\mathbb{N}$.
				\item $\max_\Baire :\subseteq \mathbb{N}^\mathbb{N} \to \mathbb{N}$ defined by $\max_\Baire(p) = \max \{p(i) \mid i \in \mathbb{N}\}$.
				\item $\operatorname{Bound} : \subseteq \mathcal{O}(\mathbb{N}) \mto \mathbb{N}$, where $n \in \operatorname{Bound}(U)$ iff $\forall m \in U \ n \geq m$.
			\end{enumerate}
		\end{lemma}

\begin{lemma}[\cite{gherardi4}]
The following are Weihrauch equivalent:
\begin{enumerate}
\item $\C_\mathbb{N}$
\item $\lim_\Delta : \subseteq \mathbb{R}^\mathbb{N} \to \mathbb{R}$, where $\lim_{\Delta}$ maps an eventually constant sequence to its limit
\end{enumerate}
\end{lemma}

\begin{lemma}[\cite{mylatz,paulymaster}]\label{lemma:lpo}
\begin{enumerate}
\item For $f : \mathbf{X} \to \{0,\ldots,n\}$ we have $\C_\mathbb{N} \nleqW f$.
\item $\lpo \leW \lpo^* \leW \C_\mathbb{N}$
\end{enumerate}
\end{lemma}

\begin{lemma}[\cite{paulybrattka}]
\label{lem:paulybrattka}
\begin{enumerate}
\item $\C_\mathbb{N} \star \C_\mathbb{N} \equivW \C_\mathbb{N}$
\item $\lim \star \C_\mathbb{N} \equivW \lim$
\end{enumerate}
\end{lemma}

\subsection{On the difference between algebraic and topological models of computation}
\name{Penrose} \cite{penrose} posed and made popular the question whether the Mandelbrot set is computable -- albeit without specifying any formal definition of \emph{computable} for this question. In the BSS-model, a negative answer was readily obtained in \cite{blum2}. \name{Brattka} \cite{brattka-emperor}, however, argued that this result just reflects that the Mandelbrot set is not an algebraic object, without being meaningful for computability as naively understood: A very similar proof applies also to the epigraph of the exponential function -- which, as many\footnote{Including \name{Penrose} himself, see e.g.~\cite[Figure 4.5, p.~167]{penrose}.} would agree, \emph{ought} to be computable.(\footnote{The question whether the (distance function of the) Mandelbrot set is computable in the computable analysis sense is still open -- as shown by \name{Hertling} \cite{hertling9}, this would be implied by the hyperbolicity conjecture. See the book \cite{braverman} for a general discussion of Julia sets and computability.})

The usual focal point for disagreement between the two communities, however, is not about functions computable in the Type-2 sense but non-computable in the BSS-sense, but vice versa. As is commonly understood, and will be proven formally below, this boils down to equality (or any other non-trivial property of real numbers) being decidable in the BSS-model. \name{Brattka} and \name{Hertling} \cite{brattkahertling} proposed the \emph{feasible real RAM}'s, a variation on the BSS-model that rather than deterministic tests can only perform non-deterministic tests that may give a wrong answer for very close numbers. The functions approximable by a feasible real RAM are precisely those computable in the Type-2 sense; i.e.~allowing approximation and removing exact tests is precisely what is needed to move from BSS-computability to Type-2 computability.

The complexity of (semi)decidable sets in the two models was compared by \name{Zhong} \cite{zhong}, and \name{Boldi} and \name{Vigna} \cite{boldi}. In \cite{zhong} it is proved that every TTE-semi-decidable set is BSS semi-decidable, and a criterion is given under which the converse holds true. In \cite{boldi} it is shown amongst other things that if an open set $U \in \mathcal{O}(\mathbb{R}^n)$ is BSS decidable with constants $c_1, \dots, c_n$, then the Turing degree of some standard name of $U$ is below the jump of the degrees of the binary expansion of the constants. Furthermore, a particular set $U$ is constructed for which the degree of every standard name of $U$ is above the jump of the degrees of the binary expansion of the constants. It is also shown that the halting set of any BSS machine with constants $c_1, \dots, c_n$ is computably (Turing-)overt relative to the constants.

We can translate some of their results into our parlance. Let $\mathrm{OD}$ denote the space of open subsets of $\mathbb{R}^*$ that are decidable by a BSS-machine using constants, represented in the obvious way by a G\"odel-number of the machine together with names for the constants. Let $\mathrm{SD}$ denote the space of Halting sets of BSS-machines using constants, again represented in the obvious way. Then:

\begin{proposition}
\begin{enumerate}
\item $\id : \mathcal{O}(\mathbb{R}^*) \to \mathrm{SD}$ is computable.
\item $\lim \leqW \left ( \id : \mathrm{OD} \to \mathcal{O}(\mathbb{R}^*) \right )$.
\item There is some $g : \subseteq \Baire \mto \mathbb{N}$ such that $\left ( \id : \mathrm{OD} \to \mathcal{O}(\mathbb{R}^*) \right ) \leqW g \star \lim$.
\item There is some $g : \subseteq \Baire \mto \mathbb{N}$ such that $\left (\overline{\phantom{A}} : \mathrm{SD} \to \mathcal{V}(\mathbb{R}^*) \right ) \leqW g$, where $\overline{\phantom{A}}$ denotes the closure operator.
\end{enumerate}
\begin{proof}
\begin{enumerate}
\item This is the uniform version of \cite[Theorem 3.1]{zhong}.
\item This follows by noting that the proof of \cite[Theorem 10]{boldi} is completely uniform.
\item This follows from \cite[Theorem 7]{boldi}.
\item This is the statement of \cite[Corollary 6]{boldi}.
\end{enumerate}
\end{proof}
\end{proposition}

While making these classifications precise, and attending to the numerous remaining questions on the complexity of sets in terms of Weihrauch reducibility seems like an interesting endeavour, we leave it to future work.


In \cite{MorozovKorovina}, the notion of semi-decidability in the TTE model is related to $\Sigma$-definability over the reals without equality: It is shown that the $\Sigma$-definable sets without equality are precisely the semi-decidable sets in the TTE-model. The main result of the paper is that there is no effective procedure which takes as input a finite formula that defines an open set with equality and takes it to a formula that defines the same open set without equality.

\subsection{Relativization}
Prima facie, several of our most general result might look unsatisfactory to the reader coming from the algebraic computation model side: We restrict our operations to be computable, hence presuppose the Type-2 notion of computability, and disallow the use of arbitrary constants that is customary for BSS-machines. These issues can be resolved directly, using the technique of \emph{relativization} from classical recursion theory. The idea here is that almost all computability-theoretic arguments remain true relative to some arbitrary, but fixed oracle $\Omega \in \Cantor$. This is of crucial importance for us due to the following:

\begin{fact}
A function is continuous iff it is computable relative to some oracle.
\end{fact}

All of our proofs relativize. Thus, for each of our results replacing each instance of \emph{computable} by \emph{continuous} again yields a true statement. Moreover, for the relativized version of a statement about an algebraic computation model, it does not change anything to allow arbitrary constants.

\section{The complexity of finitely many tests}

\subsection{Generalized register machines with $\lpo$-tests}

As explained above, the crucial distinguishing feature giving the algebraic computation models additional power is the ability to make finitely many tests, usually either equality or order. Both examples are Weihrauch equivalent to $\lpo$. Thus we are lead to the problem of classifying the computational power inherent in being allowed to make finitely many uses of $\lpo$. Note that we are not required to state any bounds in advance (which would just yield $\lpo^*$), but simply have to cease making additional queries to $\lpo$ eventually.

We can formalize this using the generalized register machines: We allow all computable functions on $\Cantor$ as functions, and $\lpo$ as test\footnote{Alternatively, we could have used the rather cumbersome Oracle-Type-Two machines suggested in \cite{paulyoracletypetwo}.}.

\begin{definition}
Let $\lpo^\diamond : \subseteq \mathbb{N} \times (\Cantor)^* \to (\Cantor)^*$ take as input a G\"odel-number of some generalized register machine $M$ on $\Cantor$ with computable functions (specified as part of the G\"odel-number) and $\lpo$-tests, as well as some input $(p_0,\ldots,p_n) \in (\Cantor)^*$ for such a machine. The output of $\lpo^\diamond$ is whatever $M$ would output on input $(p_0,\ldots,p_n)$.
\end{definition}

Of course, the preceding definition makes sense with some arbitrary multi-valued function $f$ in place of $\lpo$, and would give rise to an operation $^\diamond$ on the Weihrauch degrees. This operation is related to the \emph{generalized Weihrauch reductions} proposed by \name{Hirschfeldt} in \cite{hirschfeldt}. While a detailed investigation of $^\diamond$ seems highly desirable, it is beyond the scope of the present paper and thus relegated to future work.

\begin{proposition}\label{Proposition: C_N equiv_W lpo-diamond}
$\C_\mathbb{N} \equivW \lpo^\diamond$.
\begin{proof}
\begin{description}
\item[$\max_\Baire \leqW \lpo^\diamond$] We describe a generalized register machine program using computable functions and $\lpo$ that solves $\max_\Baire$. Given $p \in \Baire$ and $n \in \mathbb{N}$, we can compute some $p_n \in \Cantor$ such that $p_n \neq 0^\mathbb{N} \Leftrightarrow \exists i \in \mathbb{N} p(i) > n$. Starting with $i = 0$, we simply test $\lpo(p_i)$. If $\textrm{yes}$, we output $i$ and terminate. Else we continue with $i := i+1$.
\item[$\lpo^\diamond \leqW \C_\mathbb{N}$] Consider some generalized register machine with computable functions and $\lpo$-tests. Each computation path of the machine corresponding to some valid input is finite, in particular, uses $\lpo$ only finitely many times. We encode the results of the $\lpo$-tests along such a finite path by a sequence of numbers $a_1, \dots, a_m$ in the following way: If the result of the $i^{\textrm{th}}$ test is $1$, we put $a_i = 0$. If it is $0$, we put $a_i = c + 1$, where $c$ is a ``precision parameter'', intended to represent a bound on the occurrence of the first $1$ in the input to $\lpo$. Using a standard tupling function we can encode the sequence $a_1, \dots, a_m$ into a single natural number $\langle a_1, \dots, a_m \rangle$. Furthermore, we can arrange that every natural number encodes such a tuple.

Given a natural number $\langle a_1, \dots, a_m \rangle$ which encodes the results of the $\lpo$-tests along a finite path, we can check if the choices are infeasible. The number is rejected in three cases:

\begin{enumerate}
\item If the path does not end in a leaf.
\item If there exists $1 \leq i \leq m$ such that $a_i = 0$ but the input to the $i^\textrm{th}$ $\lpo$-test is non-zero.
\item If $a_i = c + 1$ but the input to the $i^\textrm{th}$ $\lpo$-test starts with more than $c$ zeroes.
\end{enumerate}

We can effectively enumerate all numbers that are rejected. This amounts to being able to compute the set $\mathrm{NR} \in \mathcal{A}(\mathbb{N})$ of numbers that are never rejected. We apply $\C_\mathbb{N}$ to this set to pick a number which is never rejected, which allows us to simulate the computation of the register machine in an otherwise computable fashion.
\end{description}
\end{proof}
\end{proposition}

On a side note, let us consider two more computational models: First, the concept of finitely revising computation presented in \cite{ziegler3} by \name{Ziegler}: These are Type-2 machines equipped with the additional power to reset their output finitely many times. It is easy to see that being computable by a finitely revising machine is equivalent to being Weihrauch-reducible to $\lim_\Delta$. Second, non-deterministic Type-2 computation, also introduced by \name{Ziegler} \cite{ziegler2} and fleshed out further by \name{Brattka}, \name{de Brecht} and \me~in \cite{paulybrattka}: Here the machine may guess an element of an advice space, and either proceed to successfully compute a solution, or reject the guess at a finite stage (and there must be a chance of the former). As shown in \cite{paulybrattka}, being computable by a non-deterministic machine with advice space $\mathbf{Z}$ is equivalent to being Weihrauch reducible to $\C_\mathbf{Z}$. We thus arrive at the following:

\begin{theorem}
The following computational models are equivalent in the sense that they yield the same class of computable functions:
\begin{enumerate}
\item Generalized register machines with computable functions and $\lpo$ as test.
\item Finitely revising machines.
\item Non-deterministic machines with advice space $\mathbb{N}$.
\end{enumerate}
\end{theorem}

We could equivalently have used generalized register machines over $\mathbb{R}$, with partial computable functions over $\mathbb{R}$ and $=$ as test. As BSS-machines are a restricted case of these, it is clear that simulating a BSS-machine is no harder than solving $\lpo^\diamond$.

\subsection{BSS-machines over the reals}

So far we have only obtained an upper bound for the power of BSS-machines in the Weihrauch lattice. For a lower bound, we require some further representatives of the Weihrauch degree of $\C_\mathbb{N}$. Let $\mathbb{Q}^{\text{e}}_+$ denote the set of non-negative rational numbers understood as a subspace of the represented space $\mathbb{R}$, i.e.~represented by the appropriate post-restriction of $\rho$. In contrast, let $\mathbb{Q}^{\text{d}}_+$ be the discrete space of non-negative rational numbers, represented by $\delta_\mathbb{Q}(0^k10^n10^m1^\omega) = \frac{n}{m+1}$. Some of the following have already been shown in \cite{mylatz,paulymaster}.
\begin{proposition}
The following are Weihrauch equivalent:
\begin{enumerate}
\item $\C_\mathbb{N}$.
\item $\max_\Baire :\subseteq \mathbb{N}^\mathbb{N} \to \mathbb{N}$.
\item $\id_{\mathbb{Q}_+}^{\text{e},\text{d}} : \mathbb{Q}^{\text{e}}_+ \to \mathbb{Q}^{\text{d}}_+$.
\item $\operatorname{Numerator} : \mathbb{Q}^{\text{e}}_+ \to \mathbb{N}$, where $\operatorname{Numerator}(q) = n$ iff $\exists m \in \mathbb{N} \ \operatorname{gcd}(n,m) = 1$ and $|q| = \frac{n}{m}$.
\item $\operatorname{Denominator} : \mathbb{Q}^{\text{e}}_+ \mto \mathbb{N}$, where $\operatorname{Denominator}(q) = m$ iff $\exists n \in \mathbb{N} \ \operatorname{gcd}(n,m) = 1$ and $|q| = \frac{n}{m}$.(\footnote{This map is multivalued, as any positive integer is a valid output on input $0$. Restricting the map to positive inputs does not change the Weihrauch degree.})
\end{enumerate}
\begin{proof}
\begin{description}
\item[$\C_\mathbb{N} \leqW \max_\Baire$] Lemma \ref{lemma:cn}.
\item[$\max_\Baire \leqW \operatorname{Denominator}$]
W.l.o.g., assume that the input $p$ to $\max_\Baire$ is monotone and that $p(0) = 0$. We proceed to compute a non-negative real number $x$ which will happen to be rational (i.e.~we compute $x$ as an element of $\mathbb{Q}^{\text{e}}_+$). Our initial approximation to $x$ is $x_0 = 0$. If $p(n+1) = p(n)$, then $x_{n+1} = x_{n}$. If $p(n+1) > p(n)$, then we search for some $k,l \in \mathbb{N}$ such that $d \left (\frac{2l + 1}{2^{\langle k, p(n+1)\rangle}}, x_n \right ) < 2^{-n-1}$ -- which are guaranteed to exist. Then we set $x_{n+1} = \frac{2l + 1}{2^{\langle k, p(n+1)\rangle}}$. As the range of $p$ is finite, this sequence will stabilize eventually, and by construction, converges quickly to its rational limit.

Applying $\operatorname{Denominator}$ to $x$ will give us some $2^{\langle k, \max_\Baire(p)\rangle} \in \mathbb{N}$. By design of $\langle \ , \ \rangle$, we can extract $\max_\Baire(p)$ from this value.

\item[$\max_\Baire \leqW \operatorname{Numerator}$] Very similar to the reduction $\max_\Baire \leqW \operatorname{Denominator}$. On monotone input $p$ with $p(0) = 0$, we start with the approximation $x_0 = 1$. If $p(n+1) = p(n)$, then $x_{n+1} = x_n$. Otherwise, we search for $k, l \in \mathbb{N}$ such that $d \left (\frac{2^{\langle p(n+1), k\rangle}}{2l + 1}, x_n \right ) < 2^{-n-1}$, and set $x_{n+1} := \frac{2^{\langle p(n+1), k\rangle}}{2l + 1}$ for these values.  As the range of $p$ is finite, this sequence will stabilize eventually, and by construction, converges quickly to its rational limit $x$.

 Applying $\operatorname{Numerator}$ to $x$ will give us some $2^{\langle k, \max_\Baire(p)\rangle} \in \mathbb{N}$. By design of $\langle \ , \ \rangle$, we can extract $\max_\Baire(p)$ from this value.

 \item[$\operatorname{Denominator} \leqW \id_{\mathbb{Q}_+}^{\text{e},\text{d}}$] Trivial.

 \item[$\operatorname{Numerator} \leqW \id_{\mathbb{Q}_+}^{\text{e},\text{d}}$] Trivial.

 \item[$\id_{\mathbb{Q}_+}^{\text{e},\text{d}} \leqW \C_\mathbb{N}$] Given a non-negative real number $x \in \mathbb{R}_+$, we can compute the closed set $\{\langle n, m\rangle \in \mathbb{N} \mid mx = n, m \neq 0\} \in \mathcal{A}(\mathbb{N})$. If $x \in \mathbb{Q}^{\text{e}}_+$ this set is non-empty, so we can use $\C_\mathbb{N}$ to extract an element, which allows us to obtain $x \in \mathbb{Q}^{\text{d}}_+$.
\end{description}
\end{proof}
\end{proposition}

\begin{proposition}
\label{prop:idqed}
$\id_{\mathbb{Q}_+}^{\text{e},\text{d}}$ is computable by a machine over $(\mathbb{R}, +, =, 1)$.
\begin{proof}
Let $x$ be the input. We can test if $x = x + x$, in which case we know that $x = 0$. If this is not the case we compute for all pairs $n, m \in \mathbb{N}$ with $n,m \geq 1$, the numbers $mx$ and $n$ by repeated addition and test them for equality. If they are equal, then we have found a valid output in the pair $(n,m)$, if not, we consider the next pair.
\end{proof}
\end{proposition}

\begin{corollary}
\label{corr:bss}
For every algebraic computation model over every structure expanding ${(\mathbb{R},+,=,1)}$ not exceeding the computable functions and $\{=,<\}$ as tests, we find that:
\begin{itemize}
\item Every partial function computable in that model is Weihrauch reducible to $\C_\mathbb{N}$.
\item There is a (partial) function computable in that model that is Weihrauch equivalent to $\C_\mathbb{N}$.
\end{itemize}
\end{corollary}

Thus, the computational power of algebraic computation models is, from the perspective of topological computation models, characterized by the Weihrauch degree of $\C_\mathbb{N}$.

\subsection{BSS-machines which compute total functions}

It should be pointed out though that the characterisation in Corollary \ref{corr:bss} relies crucially on considering partial functions, too. For total functions, we obtain instead a family of upper bounds as follows:

\begin{definition}[\cite{debrecht5}]
Let $R \subseteq \mathbb{N} \times \mathbb{N}$ be a well-founded partial order. We define $\mathfrak{L}_R : \subseteq (\mathbb{N} \times \Baire)^\mathbb{N} \to \Baire$ as follows: A sequence $(n_i,x_i)_{i \in \mathbb{N}}$ where $n_i \in \mathbb{N}$ and $x_i \in \Baire$ is in the domain of $\mathfrak{L}_R$, if $x_i \neq x_{i+1} \Rightarrow (n_{i+1}, n_{i}) \in R$ and $x_i = x_{i+1} \Rightarrow n_i = n_{i+1}$. As $R$ is well-founded, these conditions imply that the sequence stabilizes eventually, and $\mathfrak{L}_R$ returns the corresponding limit $x_\infty$.
\end{definition}

\begin{theorem}[Computable Hausdorff-Kuratowski theorem \cite{pauly-ordinals}]
For a total function $f : \mathbb{R}^* \to \mathbb{R}^*$ the following are equivalent:
\begin{enumerate}
\item $f \leqW \C_\mathbb{N}$.
\item There exists some computable $R$ such that $f \leqW \mathfrak{L}_R$.
\end{enumerate}
\end{theorem}

Let $R \subseteq \mathbb{N} \times \mathbb{N}$ and $P \subseteq \mathbb{N} \times \mathbb{N}$ be two well-founded partial orders such that there exists an order-preserving map $f$ from $R$ to $P$, i.e.~ some $f : \mathbb{N} \to \mathbb{N}$ such that if $(n,m) \in R$, then $(f(n),f(m)) \in P$. It follows that $\mathfrak{L}_R \leqW^c \mathfrak{L}_P$. Conversely, results from \cite{hertling,debrecht5} show that this implication indeed reverses. Consider well-founded partial orders up to the equivalence notion induced by the existence of order-preserving maps in both directions is one construction of the countable ordinals -- and as shown in \cite{pauly-ordinals}, in can indeed be seen as a canonical one. Thus, for any countable ordinal $\alpha$ we can associate a continuous Weihrauch degree $\mathfrak{L}_\alpha$ as the degree of $\mathfrak{L}_R$ for any/every well-founded partial order $R$ with rank $\alpha$.

\begin{proposition}[\cite{debrecht5}]
For countable ordinals $\alpha < \beta$ we have $\mathfrak{L}_\alpha \leW^c \mathfrak{L}_\beta \leW^c \C_\mathbb{N}$.
\end{proposition}

\begin{corollary}
\label{Corollary: total BSS upper bound}
Let $f : \mathbb{R}^* \to \mathbb{R}^*$ be a total BSS-computable function. Then there is some countable ordinal $\alpha$ with $f \leqW^c \mathfrak{L}_\alpha$.
\end{corollary}

\begin{proposition}
\label{Proposition: total BSS lower bound}
For each countable successor ordinal $\alpha + 1$ there is some BSS-computable total function $f_{\alpha + 1} : \uint \to \mathbb{N}$ (using constants) with $f_{\alpha + 1} \geqW^c \mathfrak{L}_{\alpha + 1}$.
\begin{proof}
We proceed by induction over $\alpha + 1$. The claim for $\alpha = 0$ is witnessed by $\lpo$. It suffices to show that if the claim is true (uniformly) for $(\alpha_i)_{i \in \mathbb{N}}$ with $\alpha_i \leq \alpha_{i+1}$, then it is true for $\left ( \sup_{i \in \mathbb{N}} \alpha_i \right ) + 1 =: \alpha + 1$.

We define $f_{\alpha + 1} : \uint \to \mathbb{N}$ piecewise. If $x \in [\frac{1}{2i+1}, \frac{1}{2i}]$ for some $i \in \mathbb{N}$, then $f_{\alpha + 1}(x) := \langle i + 1, f_{\alpha_i + 1}(2i(2i+1)x - 2i)\rangle$. Otherwise, $f_{\alpha + 1}(x) := \langle 0, 0\rangle$. If the $f_{\alpha_i + 1}$ are either uniformly (in $i$) BSS-computable, or, alternatively, coded into a real parameter, this clearly yields a BSS-computable function.

{\bf Claim: $\mathfrak{L}_{\alpha + 1} \leqW^c f_{\alpha + 1}$}

We start writing a name for $x := 0$ while reading the input to $\mathfrak{L}_{\alpha + 1}$. If the first parameter ever changes, then it will move to some $n_i$ such that the lower cone of $n_i$ has some height $\alpha' + 1 \leq \alpha$. In particular, there must be some $j \in \mathbb{N}$ with $\alpha' + 1 \leq \alpha_j + 1$. The suitable pairs $(i, j)$ can be coded into a single real parameter. There will be some $k \geq j$ such that $[\frac{1}{2k+1},\frac{1}{2k}]$ is still within the scope of the current approximation to $0$.

The tail of the input to $\mathfrak{L}_{\alpha + 1}$ at the current moment is now also a valid input to $\mathfrak{L}_{\alpha_k + 1}$. Thus, we can continue to use the reduction $\mathfrak{L}_{\alpha_k + 1} \leqW^c f_{\alpha_k + 1}$, with $f_{\alpha_k + 1}$ scaled down into the interval $[\frac{1}{2k+1},\frac{1}{2k}]$.

To interpret the output of $f_{\alpha + 1}$, the only difference is whether it is of the form $\langle 0, 0\rangle$ or $\langle i+1, n\rangle$. In the former case, the input sequence to $\mathfrak{L}_{\alpha + 1}$ is constant, and we just read off the answer. In the latter case, there is at least one change -- so we can search for it, and then split $n := \langle i', n'\rangle$ to see whether there is a further change, and so on.

\end{proof}
\end{proposition}

\subsection{BSS-machines without order tests}

If we consider total functions, and drop the order test from the signature, the maximum Weihrauch degree reachable is even lower. This is caused by two properties of algebraic sets that are not shared with semi-algebraic sets: every descending chain of algebraic sets eventually stabilises and every proper algebraic subset of an irreducible variety $V$ is nowhere dense. These properties will cause any computation tree of an algorithm which computes a total function to be finite.

\begin{proposition}
\label{Proposition: BSS with equality}
If $f : \mathbb{R}^* \to \mathbb{R}^*$ is BSS-computable over $(\mathbb{R}, +, \times, =)$, then $f \leqW^c \lpo^*$. There is a function $g: \mathbb{R}^* \to \mathbb{R}^*$ that is BSS-computable over $(\mathbb{R}, +, \times, =)$ and satisfies $g \equivW \lpo^*$.
\begin{proof}
Let $f \colon \mathbb{R}^* \to \mathbb{R}^*$ be BSS-computable over $(\mathbb{R}, +, \times, =)$. Fix some BSS-machine computing $f$. For a fixed input dimension $n$, consider the computation tree of that machine. We claim that the tree is finite. It then follows that $f \leqW^c \lpo^*$, for we can bound the number of equality tests we need to simulate the machine in terms of the size of the input tuple alone.

If the computation tree is infinite it has an infinite path by K\H{o}nig's lemma. Each node on the path corresponds to an algebraic set, and the outgoing edge from each node is labelled ``$\in$'' or ``$\notin$'', depending on whether we branch on the condition that the input is in the algebraic set or outside of it. Since the Zariski topology on $\mathbb{R}^n$ is Noetherian, there are only finitely many edges labelled with ``$\in$'' in a non-trivial way. By this we mean the following: if we number the nodes on the path with numbers $0, 1, 2, \dots$ then there exists a number $N \in \mathbb{N}$ such that for any node $v$ on the path which is labelled with a number greater than $N$ and whose outgoing edge on the path is labelled with ``$\in$'', the computation on any input will branch to ``$\in$'' if it reaches this node. Let $V$ be the algebraic set which is defined by the nodes with number smaller than $N$ whose outgoing edge on the path is labelled with ``$\in$''. Let $(V_n)_n$ be the sequence of algebraic sets which correspond to the nodes whose outgoing edge is labelled with ``$\notin$''. Our goal is to show that there exists $x \in \mathbb{R}^n$ with $x \in V$ and $x \notin V_n$ for all $n \in \mathbb{N}$. This means that $x$ passes all tests on the infinite branch (since all tests which do not correspond to $V$ or one of the $V_n$'s are passed automatically), which means that the machine runs forever on input $x$, contradicting the totality of $f$. The set $V$ is a finite union of irreducible algebraic sets, at least one of which is not contained in any of the $V_n$'s. We can hence assume without loss of generality that $V$ is itself an irreducible algebraic variety. The sets $V_n \cap V$ are (potentially empty) proper algebraic subsets of $V$ and thus have dense open complement in the Euclidean topology on $V$. By the Baire category theorem, their countable union has dense open complement in $V$. In particular, this complement contains a point. This finishes the proof.

For the example $g$, consider the function $g : \mathbb{R}^* \to \mathbb{N}$ mapping an input tuple $(x_1, \ldots, x_n)$ to the set $|\{i \in \{1,\ldots,n\} \mid x_i = 0\}|$. This function is obviously BSS-computable over $(\mathbb{R}, +, \times, =)$, and easily seen to be Weihrauch-equivalent to $\lpo^*$.
\end{proof}
\end{proposition}

\section{On the BSS-Halting problems}
Very much in analogy to the (classical) Halting problem and the investigation of the Turing degrees below it (Post's problem), the BSS-Halting problems are easily seen to be undecidable by the corresponding BSS-machines, and there are rich hierarchies to be found below them \cite{zieglermeer,gassner}. Very much unlike the classical setting, there is a natural undecidable set strictly below the Halting problem, namely $\mathbb{Q}$.

Let $\mathbb{H} \subseteq \mathbb{N} \times \mathbb{R}^*$ be the Halting problem for BSS-machines having access to $+$, $=$, and potentially additional computable operations and/or $<$ as test\footnote{Of course, changing the signature changes the set $\mathbb{H}$, but as the proof of Theorem \ref{theo:halting} is independent of these details, they do not change the Weihrauch degree of $\chi_\mathbb{H}$.}. This means that $(n, x_0,\ldots,x_n) \in \mathbb{H}$ iff $n$ is a G\"odel-number for a BSS-machine that on input $(x_0,\ldots,x_n)$ will eventually halt.

 Let $\chi_\mathbb{Q} : \mathbb{R} \to \{0,1\}$ be the characteristic function of $\mathbb{Q}$ (into the discrete space $\{0,1\}$), and $\chi_\mathbb{H} :  \mathbb{N} \times \mathbb{R}^* \to \{0,1\}$ be the characteristic function of $\mathbb{H}$. Finally, let $\isInfinite : \Cantor \to \{0,1\}$ be defined via $\isInfinite(p) = 1$ iff $|\{n \in \mathbb{N} \mid p(n) = 1\}| = \infty$.

\begin{theorem}
\label{theo:halting}
The following are Weihrauch-equivalent:
\begin{enumerate}
\item $\chi_\mathbb{Q} : \mathbb{R} \to \{0,1\}$
\item $\chi_\mathbb{H} : \mathbb{N} \times \mathbb{R}^* \to \{0,1\}$
\item $\isInfinite : \Cantor \to \{0,1\}$
\end{enumerate}
\begin{proof}
\begin{description}
\item[$\chi_\mathbb{Q} \leqW \chi_\mathbb{H}$]
Note that $\chi_\mathbb{Q} \equivW \chi_{\mathbb{Q}_+}$, where $\mathbb{Q}_+$ denotes the set of non-negative rational numbers. Now, $\chi_{\mathbb{Q}_+} \leqW \chi_\mathbb{H}$ is easily established: we just combine the original input with the program for $\id_{\mathbb{Q}_+}^{\text{e},\text{d}}$ from Proposition \ref{prop:idqed}. This will halt iff the input is rational.
\item[$\chi_\mathbb{H} \leqW \isInfinite$]
Given the G\"odel number of a BSS-machine $M$ and a standard name $p$ of a point $x \in \mathbb{R}^*$, we construct a sequence $(A_n)_n$ of Type-2 algorithms as follows: The $n^{\text{th}}$ algorithm simulates the machine $M$ on input $p$ until it reaches an equality- or order-test. It then tries to show that the result of the test is ``false'' with precision parameter $n$ (cf.~the proof of Proposition \ref{Proposition: C_N equiv_W lpo-diamond}). Otherwise it assumes that the result of the test is ``true'' and continues to simulate $M$ in this manner. If $n < m$, we say that $A_m$ \emph{refutes $A_n$ within $l$ steps} if $A_n$ and $A_m$ take different branches on the equality- or order-tests within $l$ steps of computation. This defines a relation between the numbers $m$, $n$, and $l$ which in general depends on $p$ (and not just on $x$). Note that this relation is decidable relative to $p$.

Now consider the following Type-2 algorithm: Assign a variable $n = 0$. For all pairs $(m,l) \in \mathbb{N}^2$ do the following: Run the machine $A_n$ for $l$ steps. If it doesn't halt within those $l$ steps, write a $1$ on the output tape. Test if $A_m$ refutes $A_n$ within $l$ steps. If so, put $n = m$ and write a $1$ on the output tape. Finally, write a $0$ on the output tape, and continue looping.

We claim that the sequence this algorithm produces contains finitely many $1$s if and only if $M$ halts on input $x$. Assume that the sequence contains finitely many $1$s. Then there exists $n \in \mathbb{N}$ such that $A_n$ halts within a finite number of steps and no $A_m$ refutes $A_n$ in any finite number of steps. But this means that $A_n$ correctly simulates $M$ on input $x$, for if it erroneously decides a test to be ``true'', it will eventually be refuted. Hence, $M$ halts on input $x$ within a finite number of steps. Conversely, if $M$ halts on input $x$, then it makes only finitely many equality- or order-tests before halting. Hence for sufficiently large $n$ the algorithm $A_n$ correctly simulates $M$ on input $x$, and thus is never refuted and halts after finitely many steps. It follows that our algorithm only writes finitely many $1$s on the output tape.

Hence we can decide if $M$ halts on input $x$ by applying $\isInfinite$ to the output of the algorithm.
\item[$\isInfinite \leqW \chi_\mathbb{Q}$] Compute the real number with the decimal expansion $$0.a_10a_200a_3000a_40000\ldots$$ where $a_i$ is the $i$th bit of the input. This number is rational, i.e. has a periodic decimal expansion, if and only if the $a_i$ are eventually $0$.
\end{description}
\end{proof}
\end{theorem}

This shows that the role of (local) cardinality for BSS-reducibility demonstrated in \cite{russellmiller} depends on the set of operations available in the reduction. Further, note that $\isInfinite \nleqW^c \C_\mathbb{N}$ is easily seen: While $\C_\mathbb{N}$ is $\Delta^0_2$-measurable (cf.~e.g.~\cite{paulydebrecht}), $\isInfinite$ is the characteristic function of a $\Pi^0_2$-complete set, thus not even $\Sigma^0_2$-measurable. By \cite{brattka}, levels of Borel measurability are preserved under Weihrauch reductions. Thus, the inability of BSS-machines to decide their Halting problem already holds for topological reasons (in particular, adding more continuous operations never allows a machine to solve a more restricted Halting problem).

As a side note, we shall point out that $\isInfinite$ also is the degree of deciding whether a Type-2 computable function is well-defined on some particular input:

\begin{proposition}
Let $\isDefined : \mathbb{N} \times \Cantor \to \{0,1\}$ be defined via $\isDefined(n,p) = 1$ iff the $n$-th Type-2 machine produces some $q \in \Cantor$ on input $p$. Then:
\[\isDefined \equivW \isInfinite\]
\begin{proof}
Let $s$ be an index for the Type-2 machine that copies each $1$ from the input to the output, and skips all $0$s. Then $\isInfinite(p) = \isDefined(s,p)$.

For the other direction, note that we can simulate the $n$-th Type-2 machine while writing $0$s. Whenever the machine outputs something, we write a $1$. Applying $\isInfinite$ to the output produces an answer for $\isDefined$.
\end{proof}
\end{proposition}

\section{The Weihrauch degree of sorting and other problems}
In this section we will investigate some Weihrauch degrees, in particular the degree of sorting some $p \in \Cantor$ by order. This degree will turn out to be crucial in characterizing the power of strongly analytic machines later.

Let $\mathfrak{a} = (a_n)_{n \in \mathbb{N}}$ be a computable, infinite, repetition-free and dense sequence in the complete computable metric space $\mathbf{X}$. Let $\textrm{Type}_{\mathfrak{a}} : \mathbf{X} \to \uint$ be defined via $\textrm{Type}_{\mathfrak{a}}(a_n) = 2^{-n}$ and $\textrm{Type}_{\mathfrak{a}}(x) = 0$ if $x \notin \operatorname{range}(\mathfrak{a})$.

Let $\Sort : \Cantor \to \Cantor$ be defined via $\Sort(p) = 0^n1^\omega$ iff $p$ contains exactly $n$ times the bit $0$, and $\Sort(p) = 0^\omega$ iff $p$ contains infinitely many $0$s.

\begin{theorem}
\label{theo:sorttype}
$\Sort \equivW \textrm{Type}_\mathfrak{a}$ for any $\mathfrak{a}$.
\begin{proof}
$\textrm{Type}_\mathfrak{a} \leq_{\textrm{W}} \Sort$

Let $x$ be the input to $\textrm{Type}_\mathfrak{a}$. Start testing if $x = a_0$. While this is possible, write $1$'s on the input to $\Sort$. If $x \neq a_0$ is ever proven, write a single $0$ and proceed to test if $x = a_1$ instead, while again writing $1$'s. Repeat indefinitely. The outer reduction witness is given by computable $K : \subseteq \Cantor \times \Cantor \to \uint$ with $K(p, 0^n1^\omega) = 2^{-n}$ and $K(p, 0^\omega) = 0$.

$\Sort \leqW \textrm{Type}_{\mathfrak{a}}$

As long as we find only $1$'s in the input to $\Sort$, start writing $a_0$ as the input $x$ to $\textrm{Type}_{\mathfrak{a}}$. If a $0$ is read (at time $t$), we have specified $x$ by fixing $x \in B(a_0, 2^{-t})$. As the $(a_n)_{n \in \mathbb{N}}$ are dense, we can compute an injective function $f : \mathbb{N} \to \mathbb{N}$ such that $(a_{f(n)})_{n \in \mathbb{N}}$ ranges over all $a_n$ in $B(a_0, 2^{-t})$ (and we may assume that $f(0) = 0$). We proceed to write approximations to $a_{f(1)}$ while we read $1$'s on the input to $\Sort$. If the next $0$ is read, we again compute a suitable subsequence and switch our approximations to the next element, and so on.

Given the output of $\textrm{Type}_\mathfrak{a}$, we can start by deciding whether it is $1$ or not. If it is $1$, the output of $\Sort$ must be $1^\omega$. If not, then the input to $\Sort$ must contain some $0$, and we can search until we find the first one at position $t$. Knowing $t$ means we can recover the computable function $f$ constructed in the inner reduction witness. Then we test whether the output of $\textrm{Type}_{\mathfrak{a}}$ is $2^{-f(1)}$. If so, the output of $\Sort$ is $01^\omega$. If not, there is a second $0$ somewhere in the input at time $t'$, etc.
\end{proof}
\end{theorem}

Let $\isFinite_{\mathbb{S}} : \Cantor \to \mathbb{S}$ be defined via $\isFinite_{\mathbb{S}}(p) = \top$ iff $p$ contains finitely many $1$s. Let $\isInfinite_{\mathbb{S}} : \Cantor \to \mathbb{S}$ be defined via $\isInfinite_{\mathbb{S}}(p) = \top$ iff $p$ contains infinitely many $1$s. Let $\textrm{TC}_\mathbb{N} : \mathcal{A}(\mathbb{N}) \mto \mathbb{N}$ be the total continuation of $\C_\mathbb{N}$, i.e.~$p \in \textrm{TC}_\mathbb{N}(A)$ iff $p \in A \vee A = \emptyset$.

\begin{proposition}
\label{prop:isinfiniteetc}
\hfill \begin{enumerate}
\item $\C_\mathbb{N} \leW \Sort \leW \widehat{\Sort} \equivW \lim$
\item $\isInfinite_{\mathbb{S}} \leW \textrm{TC}_\mathbb{N}$
\item $\isFinite_{\mathbb{S}} \leW \Sort$
\item $\isInfinite_{\mathbb{S}} \nleqW \Sort$
\item $\isFinite_{\mathbb{S}} \nleqW \textrm{TC}_\mathbb{N}$
\end{enumerate}
\begin{proof}
\hfill \begin{description}
\item[[Claim: $\C_\mathbb{N} \leqW \Sort$]] We show $\max_{\mathcal{O}} \leqW \Sort$ instead and appeal to Lemma \ref{lemma:cn}. Given some set $A \in \mathcal{O}(\mathbb{N})$, we start writing $1$'s. In addition, we make sure that we write exactly as many $0$s as the largest number encountered in $A$ so far. If $A$ is guaranteed to be finite, then the resulting sequence contains only finitely many $0$'s, and after it was sorted, we can read off how many there are -- which returns $\max_\mathcal{O} A$.
\item[[Claim: $\Sort \leqW \lim$]] Given some $p \in \Cantor$, let $q(\langle n, i\rangle) = 0$ if $p_{\leq i}$ contains at least $n$ $0$s, and $q(\langle n, i\rangle) = 1$ otherwise. Then $\Sort(p) = \lim(q)$.

\item[[Claim: $\widehat{\Sort} \leqW \lim$]] From $\Sort \leqW \lim$ with $\widehat{\lim} \equivW \lim$.

\item[[Claim: $\lim \leqW \widehat{\Sort}$]] From $\C_\mathbb{N} \leqW \Sort$ and $\lim \equivW \widehat{\C_\mathbb{N}}$ (see \cite[Example 3.10]{paulybrattka}).

\item[[Claim: $\lim \nleqW \Sort$]] $\Sort$ maps every input to some computable output. $\lim$ maps some computable inputs to non-computable outputs.

\item[[Claim: $\isInfinite_{\mathbb{S}} \leW \textrm{TC}_\mathbb{N}$]] Given $p \in \Cantor$, we can compute the set
\[\{n \in \mathbb{N} \mid p \textnormal{ contains no more than } n \textnormal{ 1 s}\} \in \mathcal{A}(\mathbb{N}).\] Apply $\textrm{TC}_\mathbb{N}$ to obtain some $m \in \mathbb{N}$. Now read $p$ while writing $0$s. If ever the $(m+1)$-st $1$ in $p$ is found, then the input to $\textrm{TC}_\mathbb{N}$ must have been the empty set, i.e.~$p$ must contain infinitely many $1$s. Switching the output to writing $1$s from then on causes the output to be correct.

That the reduction is strict follows from the fact that $\C_\mathbb{N} \leqW \textrm{TC}_\mathbb{N}$ (by definition), whereas $\isInfinite_{\mathbb{S}}$ has a codomain with just $2$ elements.

\item[[Claim: $\isFinite_{\mathbb{S}} \leW \Sort$]] Let $S : \Cantor \to \Cantor$ be the map that swaps $0$s and $1$s. Then $\isFinite_\mathbb{S}(p) = \delta_\mathbb{S} \circ \Sort \circ S$.

As above, that the reduction is strict follows from the fact that $\C_\mathbb{N} \leqW \Sort$ as shown above, whereas $\isFinite_{\mathbb{S}}$ has a codomain with just $2$ elements.

\item[[Claim: $\isInfinite_{\mathbb{S}} \nleqW \Sort$]] As $\Sort \leqW \lim$, we find that $\Sort$ is $\Sigma^0_2$-measurable, and so is every function reducible to it (cf.~\cite{brattka}). But $\operatorname{isInfinite}_\mathbb{S}^{-1}(\{\top\})$ is a $\Pi^0_2$-complete set.

\item[[Claim: $\isFinite_{\mathbb{S}} \nleqW \textrm{TC}_\mathbb{N}$]] Assume that $\isFinite_{\mathbb{S}} \leqW \textrm{TC}_\mathbb{N}$ via some witnesses $K'$, $H$. By composing $K'$ with $\delta_\mathbb{S}$ and $\delta_\mathbb{N}^{-1}$, we obtain computable $K : \subseteq \Cantor \times \mathbb{N} \to \mathbb{S}$. We can assume $K$ to be total, as we can let it write $0$ infinitely many times without changing the value of the output. Note that we can assume that this procedure for making $K$ total preserves single-valuedness since $\mathbb{N}$ is discrete. We can turn $H$ into a computable function $h : \{0,1\}^* \to \mathbb{N}^*$ such that $n \notin \psi(H(p))$ iff $\exists l. \ (n + 1) \in h(p_{\leq l})$, where $p_{\leq l}$ is the prefix of $p$ of length $l$. We will reason with $K$ and $h$ in the following.

    As $\isFinite_{\mathbb{S}}(0^\omega) = \top$, there have to be $l_0, k \in \mathbb{N}$ such that $K(0^{l_0}\{0,1\}^\omega, k) = \top$. As $\isFinite_{\mathbb{S}}(0^{l_0}1^\omega) = \bot$, there has to be some $l_1 \in \mathbb{N}$ such that $h(0^{l_0}1^{l_1}) \ni k + 1$.

    Now we proceed in stages $i \in \mathbb{N}$, each associated with some current prefix $p_i \in \{0,1\}^*$. We start with $i = 0$ and $p_0 = 0^{l_0}1^{l_1}$. In stage $i$, consider $K(p_i0^\omega, i)$. If this is $\bot$, then we must have $i \notin \psi(H(p_i0^\omega))$ (otherwise the reduction could answer $\bot$ wrongly), so there is some $j$ such that $h(p_i0^{j}) \ni i + 1$. We set $p_{i+1} := p_i0^{j}1$ and continue with the next stage. If $K(p_i0^\omega, i) = \top$, then this is already determined by some finite prefix $p_i0^{j}$. Hence, we must have that $i \notin \psi(H(p_i0^{j}1^\omega))$ (otherwise the reduction could answer $\top$ wrongly), so there is some $m$ with $h(p_i0^{j}1^{m}) \ni i+1$. We set $p_{i+1} := p_i0^{j}1^{m}1$ and continue with the next stage.

    We find that $p = \lim_{i \to \infty} p_i \in \Cantor$ is well-defined, contains infinitely many $1$s and satisfies that $\psi(H(p)) = \emptyset$. Thus, a realiser of $\textrm{TC}_\mathbb{N}$ may answer anything on input $H(p)$, in particular it may answer $k$. But $K(p,k) = \top$, so the reduction fails.

\item[[Claim: $\Sort \nleqW \C_\mathbb{N}$]] By combining $\isFinite_{\mathbb{S}} \leW \Sort$ and $\isFinite_{\mathbb{S}} \nleqW \textrm{TC}_\mathbb{N}$, and noting that $\C_\mathbb{N} \leqW \textrm{TC}_\mathbb{N}$ by definition.

\end{description}
\end{proof}
\end{proposition}

We next wish to show that $\Sort \leW \Sort^2 \leW \Sort^3 \leW \ldots$. For this, we need a slight generalization of the Squashing Theorem from \cite{shafer} (cf.~the development in \cite{tahina-phd}). Rather than using the notion of a finitely-tolerant function as employed there, we generalize this to weakly finitely-tolerant functions. The proof remains unaffected by this, though.

\begin{definition}
Call $f : \Baire \mto \Baire$ weakly finitely-tolerant, if there is a computable function $A$ such that for any $\lambda, \lambda' \in \mathbb{N}^*$, $p,q \in \Baire$, if $q \in f(\lambda p)$, then $A(q,\lambda,\lambda') \in f(\lambda' p)$.
\end{definition}

\begin{theorem}[Squashing Theorem \cite{shafer}]
If $f : \Baire \mto \Baire$ is weakly finitely tolerant and $f \equivW f \times f$, then $f \equivW \widehat{f}$.
\end{theorem}

\begin{corollary}
\label{Corollary: Sort^n < Sort^(n + 1)}
For any $n \in \mathbb{N}$, $\Sort^n \leW \Sort^{n+1}$.
\begin{proof}
It is easy to see that any $\Sort^i$ is weakly finitely tolerant. Thus, if the claim were false, the Squashing Theorem would imply $\Sort^i \equivW \widehat{\Sort} \equivW \lim$ (from Proposition \ref{prop:isinfiniteetc} (1)), but the left-hand side has only computable outputs, whereas the right-hand side maps some computable inputs to non-computable outputs.
\end{proof}
\end{corollary}

\begin{proposition}
\begin{enumerate}
\item $\Sort \star \lpo \equivW \Sort \times \lpo$
\item $\Sort \star \C_\mathbb{N} \equivW \Sort \times \C_\mathbb{N}$
\end{enumerate}
\begin{proof}
Note that we can extend any computable partial function $f: \subseteq \Baire \to \Cantor$ to a computable total function $F: \Baire \to \Cantor$ such that $\Sort \circ f (p) = \Sort \circ F (p)$ for any $p \in \dom(f)$ -- just add infinitely many $1$'s to the (partial) output of $f$. Furthermore, for both cases it suffices to show the $\leqW$-direction, the $\geqW$-direction trivially holds.
\begin{enumerate}
\item As long as the input to $\lpo$ is consistent with $0^\mathbb{N}$, we use the corresponding input to $\Sort$. If we ever read a $1$ in the input to $\lpo$, we restart writing the input to $\Sort$ corresponding to the now known output of $\lpo$. By looking at the output of $\lpo$ (on the right hand side), we can determine in which case we are. In the former, the output of $\Sort$ already is correct. In the latter, we can find out the precise finite prefix of the input to $\Sort$ we had written when encountering the $1$ in the input to $\lpo$, count the $0$s in that prefix and adjust the output of $\Sort$ accordingly.
\item We use $\max_\Baire$ instead of $\C_\mathbb{N}$. The argument proceeds similar as above: Start by providing the input to $\Sort$ that would correspond to $\max_\Baire$ outputting $0$. Once we learn that $\max_\Baire$ will provide a larger value, switch to the corresponding input to $\Sort$. Repeat as required. We can then use $\max_\Baire$ (on the right hand side) to determine the length of the finite wrong prefix fed to $\Sort$, and change the output of $\Sort$ accordingly to fix it.
\end{enumerate}
\end{proof}
\end{proposition}

\begin{proposition}
\label{prop:sortcnisinfinite}
$\textrm{TC}_\mathbb{N} \leqW \C_\mathbb{N} \star \isFinite_{\mathbb{S}}$ and $\Sort \leqW \C_\mathbb{N} \star \isInfinite_{\mathbb{S}}$.
\begin{proof}
{\bf First claim}: Given some $A \subseteq \mathbb{N}$, we can compute $p_A \in \Cantor$ such that $|\{k \in \mathbb{N} \mid p_A(k) = 1\}| \geq n$ iff $\{0,\ldots,n-1\} \cap A = \emptyset$. Apply $\isFinite_{\mathbb{S}}$ to this, and then $\left (\id : \mathbb{S} \to \{0,1\} \right ) \equivW \lpo$ to the output. If the answer is $1$, the original input is a valid input for $\C_\mathbb{N}$. If we learn $0$ from the first part, feed a name for $\mathbb{N}$ to $\C_\mathbb{N}$. The output of $\C_\mathbb{N}$ yields a solution to $\textrm{TC}_\mathbb{N}$ in either case. Thus, $\textrm{TC}_\mathbb{N} \leqW \C_\mathbb{N} \star \lpo \star \isFinite_{\mathbb{S}} \equivW \C_\mathbb{N} \star \isFinite_{\mathbb{S}}$. Here we used $\lpo \leqW \C_\mathbb{N}$, and $\C_\mathbb{N} \star \C_\mathbb{N} \equivW \C_\mathbb{N}$ (as recalled in Lemma \ref{lem:paulybrattka}).

{\bf Second claim}: We show that $\Sort \leqW \max_\Baire \star \lpo \star \isInfinite_{\mathbb{S}}$ instead. Let $S : \Cantor \to \Cantor$ swap $0$ and $1$ componentwise, and let $p \in \Cantor$ be the original input to $\Sort$. Apply $\id : \mathbb{S} \to \{0,1\}$ to the output of $\isInfinite_{\mathbb{S}}$ on input$S(p)$. If we receive a $1$, feed $0^\mathbb{N}$ to $\max$, and answer $0^\mathbb{N}$ for $\Sort$. Else, define $q \in \Baire$ by $q(n) = | \{k \leq n \mid p(k) = 0\}|$, and note that $q \in \dom(\max_\Baire)$. Then $0^{\max q}1^\omega$ is the correct output to $\Sort$.
\end{proof}
\end{proposition}

\begin{corollary}
\label{corr:bigcorr}
$\C_\mathbb{N} \star \Sort \equivW \C_\mathbb{N} \star \isFinite_{\mathbb{S}} \equivW \C_\mathbb{N} \star \isInfinite_{\mathbb{S}} \equivW \C_\mathbb{N} \star \textrm{TC}_\mathbb{N}$
\begin{proof}
\begin{description}
\item[$\C_\mathbb{N} \star \isInfinite_{\mathbb{S}} \leqW \C_\mathbb{N} \star \textrm{TC}_\mathbb{N}$] By Proposition \ref{prop:isinfiniteetc} (2).
\item[$\C_\mathbb{N} \star \textrm{TC}_\mathbb{N} \leqW \C_\mathbb{N} \star \isFinite_{\mathbb{S}}$] By Proposition \ref{prop:sortcnisinfinite} we have $\C_\mathbb{N} \star \textrm{TC}_\mathbb{N} \leqW \C_\mathbb{N} \star \C_\mathbb{N} \star \isFinite_{\mathbb{S}} \equivW \C_\mathbb{N} \star \isFinite_{\mathbb{S}}$, invoking Lemma \ref{lem:paulybrattka}.
\item[$\C_\mathbb{N} \star \isFinite_{\mathbb{S}} \leqW \C_\mathbb{N} \star \Sort$] By Proposition \ref{prop:isinfiniteetc} (3).
\item[$\C_\mathbb{N} \star \Sort \leqW \C_\mathbb{N} \star \isInfinite_{\mathbb{S}}$] By Proposition \ref{prop:sortcnisinfinite} we have $\C_\mathbb{N} \star \Sort \leqW \C_\mathbb{N} \star \C_\mathbb{N} \star \isInfinite_{\mathbb{S}} \equivW \C_\mathbb{N} \star \isInfinite_{\mathbb{S}}$.
\end{description}
\end{proof}
\end{corollary}

Our next goal is to show that $\Sort$ is 2-low, in the sense that $\lim \star \lim \star \Sort \equivW \lim \star \lim$. We start with a more general result, for which we need the notion of a precomplete represented space. A space $\mathbf{Y}$ is called precomplete, if for every partial computable function $F : \subseteq \Cantor \to \mathbf{Y}$ there is a total computable function $F' : \Cantor \to \mathbf{Y}$ such that $F = F'|_{\dom(F)}$. Typical examples of precomplete spaces are $\mathbb{S}$, $\mathcal{O}(\mathbb{N})$ and $\mathcal{A}(\mathbb{N})$.

We further use the precomplete space $\mathbb{S}_{\Sigma^0_2}$ with underlying set $\{\bot, \top\}$ and representation $\delta_{\Sigma^0_2} : \Cantor \to \{\bot,\top\}$ defined via $\delta_{\Sigma^0_2}(p) = \top$ iff $p$ contains infinitely many $1$s, and $\delta_{\Sigma^0_2}(p) = \bot$ else. Now the map $\id : \mathbb{S}_{\Sigma^0_2} \to \{0,1\}$ (mapping $\bot$ to $0$ and $\top$ to $1$) has the same realizer as $\isInfinite$. Moreover, $\widehat{\isInfinite} \equivW \lim \star \lim$ was shown in \cite{brattka}.

\begin{proposition}
\label{prop:compparr}
Let $g : \mathbf{X} \mto \mathbb{N}$, and $f : \mathbf{Y} \mto \mathbf{Z}$ where $\mathbf{Y}$ is precomplete. Then $f \star g \leqW \widehat{f} \times g$.
\begin{proof}
We make use of the explicit representative of $f \star g$ constructed in \cite{paulybrattka4}. The input of $f \star g$ is a partial continuous function $e : \subseteq \mathbb{N} \mto \Baire \times \mathbf{Y}$ and some $x \in \mathbf{X}$. The output is a pair $(p,z) \in \Baire \times \mathbf{Z}$ such that there exists $y \in \mathbf{Y}$ such that $(p,y) \in e(g(x))$ and $z \in f(y)$. As $\mathbb{N}$ has an injective representation, we can assume w.l.o.g.~that $e$ is actually a single-valued (partial) function, which then splits into $e_1 : \subseteq \mathbb{N} \to \Baire$ and $e_2 : \subseteq \mathbb{N} \to \mathbf{Y}$. As $\mathbf{Y}$ is precomplete, we can extend $e_2$ to a total function, and by currying, obtain a sequence $(y_0,y_1,\ldots)$. We apply $\widehat{f}$ to $(y_0,y_1,\ldots)$ and obtain some sequence $(z_0,z_1,\ldots,)$, and $g$ to $x$ to obtain $n \in \mathbb{N}$. The pair $(e_1(n), z_n)$ constitutes a valid output for $f \star g$.
\end{proof}
\end{proposition}

\begin{corollary}
\label{corr:low2}
$\lim \star \lim \star \Sort \equivW \lim \star \lim$
\begin{proof}
From Proposition \ref{prop:sortcnisinfinite} we can in particular conclude that $\Sort \leqW \C_\mathbb{N} \star \isInfinite$, and the righthand side has up to isomorphism codomain $\mathbb{N}$. The Weihrauch degree of $\lim \star \lim$ has $\widehat{\left ( \id: \mathbb{S}_{\Sigma^0_2} \to \{0,1\} \right )}$ as a representative with precomplete domain. Thus, from Proposition \ref{prop:compparr} we conclude: \[\lim \star \lim \star \Sort \leqW (\lim \star \lim) \times (\C_\mathbb{N} \star \isInfinite)\] Since $(\C_\mathbb{N} \star \isInfinite) \leqW \C_\mathbb{N} \star (\lpo \star \isFinite_\mathbb{S}) \leqW (\C_\mathbb{N} \star \lpo) \star \Sort \leqW \lim \star \lim$ using Proposition \ref{prop:isinfiniteetc} (1,3), Lemma \ref{lem:paulybrattka} (1), and Lemma \ref{lemma:lpo} (2), the claim follows.
\end{proof}
\end{corollary}

\begin{corollary}
$\coprod_{n \in \mathbb{N}} \Sort^{(n)} \leqW \lim \star \lim$.
\begin{proof}
By iterating Corollary \ref{corr:low2} $n$ times, we find that: \[\lim \star \lim \star \Sort^{(n)} \leqW \lim \star \lim\] As this argument is uniform in $n$, the claim follows.
\end{proof}
\end{corollary}

\section{The algebraic decision problem}
\label{sec:thedecisionproblem}
We are now ready to introduce and study a canonical problem associated with strongly analytic machines. Let $(P_{n,d})_{n \in \mathbb{N}}$ be some standard enumeration of the $d$-variate polynomials with rational coefficients.

\begin{definition}
Define functions $\textsc{AlgDec} : \mathbb{R}^* \to [0, 1]$ and $\textsc{AlgDec}_d : \mathbb{R}^d \to [0, 1]$ via $$\textsc{AlgDec}_d(x_1, \ldots, x_d) = \sum_{\{n \mid P_{n, d}(x_1, \ldots, x_d) = 0\}} 2^{-2n-2}$$ and $\textsc{AlgDec}(x_1, \ldots, x_m) = \textsc{AlgDec}_m(x_1, \ldots, x_m)$.
\end{definition}

The choice for $\uint$ as the codomain for $\textsc{AlgDec}$ is just to ease the comparison to functions computable by strongly analytic machines, we could just as well have defined $\textsc{AlgDec}_d : \mathbb{R}^d \to \Cantor$ with $\textsc{AlgDec}_d((x_1,\ldots,x_d))(n) = 1$ iff $P_{n,d}(x_1,\ldots,x_d) = 0$. Thus, intuitively $\textsc{AlgDec}$ will tell us the ``algebraic type'' (in the sense of the definition before Theorem \ref{theo:sorttype}) of the input tuple.

\begin{observation}
\label{obs:analytic}
If $f : \mathbb{R}^* \to \mathbb{R}^*$ is computable by a strongly analytic machine, then $f \leqW \textsc{AlgDec}$. Moreover, $\textsc{AlgDec}$ is computable by a strongly analytic machine.
\begin{proof}
That $\textsc{AlgDec}$ is computable by a strongly analytic machine is immediate. For the remainder of the claim, we argue that if a strongly analytic machine $M$ computes $f : \mathbb{R}^* \to \mathbb{R}^*$ on input $\mathbf{x}$, then a Type-2 machine can simulate $M$ if provided $\mathbf{x}$ and $\textsc{AlgDec}(\mathbf{x})$ as input. Clearly, the only obstacle to such a simulation are the equality tests that $M$ can make. Each of these is of the form $p(\mathbf{x}) = 0?$, where $p$ is a rational multivariate polynomial\footnote{If $M$ is using real constants, we would consider these as part of $\mathbf{x}$.}. If $p = P_{n,d}$, then inspecting $\textsc{AlgDec}(\mathbf{x})$ up to precision $2^{-2n-4}$ allows the Type-2 machine to determine whether or not $p(\mathbf{x}) = 0$.
\end{proof}
\end{observation}

\begin{observation}
$\textsc{AlgDec} \equivW \left ( \coprod_{d \in \mathbb{N}} \textsc{AlgDec}_d \right )$
\end{observation}

\begin{proposition}\label{prop: algdec1 equiv sort}
$\textsc{AlgDec}_1 \equivW \Sort$
\begin{proof}
Let $\mathfrak{a}$ be an effective enumeration of the algebraic numbers in $\mathbb{R}$. We understand this to mean that an index $n$ of an algebraic number $a_n$ encodes the minimal polynomial of $a_n$, together with some information about which root (e.g.~ordered by $<$) $a_n$ is of its minimal polynomial. By Theorem \ref{theo:sorttype} we have that $\Sort \equivW \textsc{Type}_{\mathfrak{a}}$, thus we only need to show $\textsc{AlgDec}_1 \equivW \textsc{Type}_\mathfrak{a}$.

For $\textsc{AlgDec}_1 \leqW \textsc{Type}_\mathfrak{a}$ we show that for a given rational polynomial $P$ the predicate $P(x) = 0$ is decidable relative to $\textsc{Type}_\mathfrak{a}(x)$. Given a non-zero rational polynomial $P$, we verify in parallel if $P(x) \neq 0$ and if $\textsc{Type}_\mathfrak{a}(x) \neq 0$. Clearly, one of the searches has to terminate. If the second search terminates, it yields the minimal polynomial of $x$. Now we can decide if $P(x) = 0$ by deciding if the minimal polynomial divides $P$.

For $\textsc{Type}_{\mathfrak{a}} \leqW \textsc{AlgDec}_1$, observe that from any non-zero rational polynomial $P$ and a real number $x \in \mathbb{R}$ with $P(x) = 0$, we can compute the minimal polynomial of $x$ and determine the  position of $x$ in the list of its roots.
\end{proof}
\end{proposition}

\begin{theorem}\label{theo: algdec}
$\textsc{AlgDec}_d \equivW \textsc{AlgDec}_1^d$
\end{theorem}

\begin{corollary}
\label{corr:algdacsortstar}
$\textsc{AlgDec} \equivW \Sort^*$
\end{corollary}

In order to prove Theorem \ref{theo: algdec} we need to recall a few facts from (computational) commutative algebra. Let $\IQ$ denote the represented space of ideals in $\QP$, where an ideal is represented by some finite set of generators (that this is a representation follows from Hilbert's basis theorem). Recall that the \emph{height} $\height(P)$ of a prime ideal $P$ is the length $n$ of the longest chain of strict inclusions
\[P = P_n \supsetneq P_{n - 1} \supsetneq \dots \supsetneq P_1 \supsetneq P_0,  \]
where the $P_i$'s are prime ideals. The \emph{Krull dimension} of a ring $R$ is the supremum of the heights of all prime ideals in $R$. If $K$ is a field, then the polynomial ring $K[X_1,\dots, X_d]$ has Krull dimension $d$. We will need the following well-known facts from computer algebra (see e.g. \cite{Singular, Kalkbrener, PrimeIdeals, Buchberger1, Buchberger2})

\begin{fact}\label{fact: computability of ideal operations}\hfill
\begin{enumerate}
\item Membership of a polynomial $f \in \QP$ in an ideal $I \in \IQ$ is decidable.
\item Primality of a given ideal $I \in \IQ$ is decidable.
\item The height of a given prime ideal $P \in \IQ$ is computable.\qed
\end{enumerate}
\end{fact}

\begin{proof}[Proof of Theorem \ref{theo: algdec}]
The direction $\textsc{AlgDec}_1^d \leqW \textsc{AlgDec}_d$ is trivial. For the converse direction we prove $\textsc{AlgDec}_d \leq_W \Sort^d$ and apply Proposition \ref{prop: algdec1 equiv sort}. For a given point $x \in \mathbb{R}^d$, consider the prime ideal $I(x) = \{f \in \QP \mid f(x) = 0 \}$. Our goal is to compute the characteristic function of $I(x)$. We will use $\Sort^d$ to approximate $I(x)$ ``from below'' in the following sense: we will compute a sequence $(P_n)_n$ of prime ideals with $P_n \subseteq I(x)$ for all $n$ and $P_n = I(x)$ for sufficiently large $n$. Using the sequence $(P_n)_n$ we can verify if $f \in I(x)$ by searching for an $n \in \mathbb{N}$ such that $f \in P_n$, using Fact \ref{fact: computability of ideal operations} (1). Conversely, we can verify if $f \notin I(x)$ by verifying if $f(x) \neq 0$.

It remains to construct the sequence $(P_n)_n$. By Fact \ref{fact: computability of ideal operations} we can compute for each $h \in [1;d]$ a sequence $(P_{h,n})_n$ containing all prime ideals in $\QP$ of height $\geq h$. For each of these sequences we use an instance of $\Sort$ to compute a new sequence $(P'_{h,n})_n$ with $P'_{h,n} \subseteq I(x)$, proceeding like in the proof of the reduction $\textrm{Type}_{\mathfrak{a}} \leq_W \Sort$ in Theorem \ref{theo:sorttype}: start with $P_{h,0}$ and try to prove that $P_{h,0} \not \subseteq I(x)$ by searching for a generator $f$ of $P_{h,0}$ with $f(x) \neq 0$. At the same time, write $1$'s to the input of $\Sort$. If $P_{h,0} \not \subseteq I(x)$ is proved, write a $0$ and continue with $P_{h,1}$. Apply $\Sort$ to the resulting sequence to obtain a new sequence $p \in \Cantor$. If $p = 0^N1^\omega$, put
\[P'_{h,n} = \begin{cases} \langle 0 \rangle &\text{if }n \leq N, \\
						   P_{h,N} &\text{otherwise.} \end{cases} \]
If $p = 0^{\omega}$, put $P'_{h,n} = \langle 0 \rangle$ for all $n \in \mathbb{N}$.						

By construction, each of the ideals $P'_{h,n}$ is contained in $I(x)$. If $h \leq \height(I(x))$, then there exists $n \in \mathbb{N}$ such that $P'_{h,n}$ has height $\geq h$. In particular, if $h = \height(I(x))$, then there exists $P'_{h,n} \subseteq I(x)$ with $\height(P'_{h,n}) \geq \height(I(x))$. Since $P'_{h,n}$ is prime it follows that $P'_{h,n} = I(x)$. Since $\QP$ has Krull dimension $d$, the height $\height(I(x))$ is a number between $0$ and $d$, and if $\height(I(x)) = 0$, then $I(x) = \langle 0 \rangle$, so that $P'_{h,n} = I(x)$ for all $n \in \mathbb{N}$, $h \in [1;d]$. In any case, there always exist $h$ and $n$ such that $P'_{h,n} = I(x)$. Using standard coding tricks we may write the double-sequence $(P'_{h,n})_{h \in [1;d], n \in \mathbb{N}}$ as a single sequence $(J_n)_n$. The $J_n$'s are a sequence of prime ideals contained in $I(x)$, at least one of which is equal to $I(x)$. Now put $P_n = \sum_{k = 0}^n J_k$. Then $P_n \subseteq I(x)$ for all $n$ and $P_n = I(x)$ for sufficiently large $n$.
\end{proof}

\cite[Question 3.9]{zieglergaertner} asks whether there is a set $A \subseteq \mathbb{R}$ which is \emph{weakly semidecidable}, yet not a $\Pi^0_3$-set. They define a set to be weakly semidecidable, if it is BSS many-one reducible to the boundedness problem for analytic machines. This in turn means that there is a BSS-computable function $H : \mathbb{R} \to \mathbb{R}^*$, and an analytic machine that on input $H(x)$ for $x \in A$ computes some bounded sequence $(a_i) \in \mathbb{R}^\mathbb{N}$, and on input $H(x)$ for $x \notin A$ computes some unbounded sequence $(a_i) \in \mathbb{R}^\mathbb{N}$. We can give a negative answer:

 \begin{proposition}
 Every weakly semidecidable set is $\Delta^0_3$.
 \begin{proof}
Given some weakly semidecidable set $A$, we want to provide an upper bound on the Weihrauch degree of $\chi_A : \mathbb{R} \to \{0,1\}$. By combining Corollary \ref{corr:bss}, Observation \ref{obs:analytic} and Corollary \ref{corr:algdacsortstar}, we see that there is a function $a : \mathbb{R} \to \mathbb{R}^\mathbb{N}$ with $a \leqW \Sort^* \star \C_\mathbb{N}$, such that $a(x)$ is bounded iff $x \in A$.

Now given $(a_i) \in \mathbb{R}^\mathbb{N}$, we can compute some $p \in \Cantor$ such that $p$ contains at least $n$ $0$s iff $\exists k \ |a_k| > n$. Thus, using $\isInfinite$, we can detect whether a real sequences is bounded or unbounded.

Put together, we conclude that $\chi_A \leqW \isInfinite \star \Sort^* \star \C_\mathbb{N}$. Since $\isInfinite \leqW \lim \star \lim$, we have $\chi_A \leqW \lim \star \lim \star \Sort^* \star \C_\mathbb{N}$. By Corollary \ref{corr:low2} the righthand side is reducible to
$\lim \star \lim \star \C_\mathbb{N}$. With Lemma \ref{lem:paulybrattka} we conclude
$\chi_A \leqW \lim \star \lim$. As $\lim \star \lim$ is $\Sigma^0_3$-measurable, the claim follows.
 \end{proof}
 \end{proposition}

\section{Comparing the SCI in the two models}
Following \cite{hansen,hansen2} we shall define the \emph{solvability complexity index} over the BSS-model and over the TTE-model, and then use the results on Weihrauch degrees obtained in the preceding sections to bound their difference.

\begin{definition}
An $n$-tower for a function $f : \mathbb{R}^* \to \mathbb{R}^*$ is a function $F : \mathbb{N}^n \times \mathbb{R}^* \to \mathbb{R}^*$ such that $f(x_1,\ldots,x_m) = \lim_{i_1 \to \infty} \ldots \lim_{i_n \to \infty} F(i_1,\ldots,i_n,x_1,\ldots,x_m)$. For some function $f : \mathbb{R}^* \to \mathbb{R}^*$, let $\textrm{SCI}_{\textrm{BSS}}(f)$ be the least $n$ such that there a BSS-computable $n$-tower for $f$. Let $\textrm{SCI}_{\textrm{TTE}}(f)$ be the least $n$ such that there a computable (i.e.~TTE-computable) $n$-tower for $f$.
\end{definition}

\begin{observation}
$\textrm{SCI}_{\textrm{TTE}}(f) \leq n$ iff $f \leqW \lim^{(n)}$.
\end{observation}

\begin{theorem}
\label{theo:sci}
If $\textrm{SCI}_{\textrm{TTE}}(f) \geq 2$ or $\textrm{SCI}_{\textrm{BSS}}(f) \geq 2$, then $\textrm{SCI}_{\textrm{TTE}}(f) = \textrm{SCI}_{\textrm{BSS}}(f)$.
\begin{proof}
{\bf Assume $\textrm{SCI}_{\textrm{TTE}}(f) = n \geq 1$}. Let $F : \mathbb{N}^n \times \mathbb{R}^* \to \mathbb{R}^*$ be a computable $n$-tower for $f$. By the Stone-Weierstrass Approximation Theorem, we can approximate $F$ by rational polynomials on each hypercube, i.e.~there are rational multivariate polynomials $g_{i_1,\ldots,i_n}^k$ such that for $(x_1,\ldots,x_m) \in [-k,k]^m$ we find that $d(g_{i_1,\ldots,i_n}^k(x_1,\ldots,x_m), F(i_1,\ldots,i_n,x_1,\ldots,x_m)) < 2^{-k}$. As we can code a computable countable sequence of rational multivariate polynomials into a computable parameter, we find that $G : \mathbb{N}^n \times \mathbb{R}^* \to \mathbb{R}^*$ defined via $G(i_1,\ldots,i_n,x_1,\ldots,x_m) = g^{i_n}_{i_1,\ldots,i_n}(x_1,\ldots,x_m)$ is BSS-computable. Moreover, it is straight-forward to verify that $G$ also is an $n$-tower for $f$. Thus, $\textrm{SCI}_{\textrm{TTE}}(f) \geq \textrm{SCI}_{\textrm{BSS}}(f)$.

{\bf Assume $\textrm{SCI}_{\textrm{BSS}}(f) = n \geq 2$}. Let $F$ be a BSS-computable $n$-tower for $f$. We curry $F$ to obtain $G : \mathbb{R}^* \to (\mathbb{R}^*)^{\mathbb{N}^n}$, and notice that with the same reasoning as for strongly analytic functions in Section \ref{sec:thedecisionproblem}, we find that $G \leqW \Sort^*$. By assumption, we then have that $f \leqW \lim^{(n)} \star \Sort^*$. As $n \geq 2$, Corollary \ref{corr:low2} implies that already $f \leqW \lim^{(n)}$. By \cite[Fact 5.5]{gherardi4} $\lim$ is a \emph{transparent cylinder}, i.e.~it follows that there is a computable function $H$ such that $f = \lim \circ \ldots \circ \lim \circ H$. But this means that $H$ is a computable $n$-tower for $f$, i.e.~$\textrm{SCI}_{\textrm{BSS}}(f) \geq \textrm{SCI}_{\textrm{TTE}}(f)$.
\end{proof}
\end{theorem}

\begin{corollary}
For $n \geq 2$, $\textrm{SCI}_{\textrm{BSS}}(f) \geq n$ iff $f \nleqW \lim^{(n-1)}$.
\end{corollary}

Theorem \ref{theo:sci} provides a formal version of the informal idea that a function that is \emph{very non-computable} in the BSS-model is still so in the TTE-model and vice versa. This provides a reason to continue the investigation of the non-computability of an interesting function beyond establishing it -- as a more precise classification can potentially be translated to the other setting via Theorem \ref{theo:sci}. This also shows that there is common ground between the two frameworks, and that this common ground includes the SCI (provided it is at least $2$).

\section{Summary diagram}
The following diagram provides an overview of some the relevant Weihrauch degrees. Arrows denote reductions in the reverse direction. The diagram is complete up to transitivity, i.e.~if no arrow is present in the transitive closure of the diagram, then there is a separation proof for the principles.

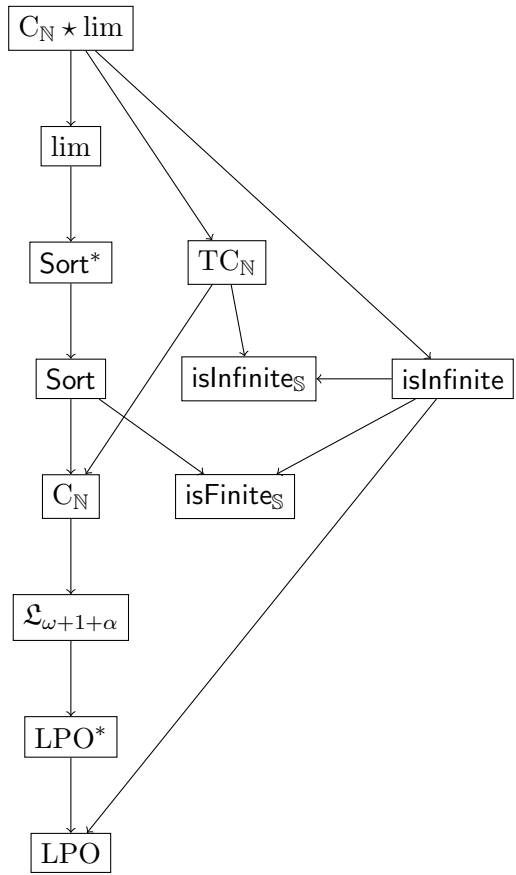
\begin{minipage}{0.5\textwidth}
\begin{center}
\begin{tikzpicture}[node distance=1cm, auto]
  \node[draw] (cnlim) {$\C_\mathbb{N} \star \lim$};
  \node[draw] (lim) [below=of cnlim] {$\lim$};
  \node[draw] (sortstar) [below=of lim] {$\Sort^*$};
  \node[draw] (sort) [below=of sortstar] {$\Sort$};
  \node[draw] (CN) [below=of sort] {$\C_\mathbb{N}$};
  \node[draw] (lalpha) [below=of CN] {$\mathfrak{L}_{\omega+1+\alpha}$};
  \node[draw] (lpostar) [below=of lalpha] {$\lpo^*$};
  \node[draw] (lpo) [below=of lpostar] {$\lpo$};
  \node[draw] (tcn) [right=of sortstar] {$\mathrm{TC}_\mathbb{N}$};
  \node[draw] (isFinite) [right=of CN] {$\isFinite_\mathbb{S}$};
  \node[draw] (isInfiniteS) [right=of sort] {$\isInfinite_\mathbb{S}$};
  \node[draw] (isInfinite) [right=of isInfiniteS] {$\isInfinite$};

  \draw[->] (cnlim) to node {} (lim);
  \draw[->] (lim) to node {} (sortstar);
  \draw[->] (sortstar) to node {} (sort);
    \draw[->] (sort) to node {} (CN);
    \draw[->] (CN) to node {} (lalpha);
  \draw[->] (lalpha) to node {} (lpostar);
  \draw[->] (lpostar) to node {} (lpo);
   \draw[->] (sort) to node {} (isFinite);
    \draw[->] (cnlim) to node {} (tcn);
     \draw[->] (tcn) to node {} (CN);
     \draw[->] (tcn) to node {} (isInfiniteS);
      \draw[->] (isInfinite) to node {} (isInfiniteS);
       \draw[->] (isInfinite) to node {} (isFinite);
        \draw[->] (isInfinite) to node {} (lpo);
        \draw[->] (cnlim) to node {} (isInfinite);
\end{tikzpicture}
\end{center}
\end{minipage}
\begin{minipage}{0.45\textwidth}
\begin{itemize}
\item $\lim$ captures limit computability.
\item $\Sort^*$ captures computability by strongly analytic machines
\item $\isInfinite$ contains BSS-Halting problems
\item $\C_\mathbb{N}$ captures computability by BSS-machines.
\item $\lpo^*$ captures computability of total functions by BSS-machines without $<$
\end{itemize}
\end{minipage}

\bibliographystyle{eptcs}
\bibliography{models}

\section*{Acknowledgements}
This work was inspired by discussions at the workshop \emph{Real Computation and BSS Complexity} in Greifswald, and the second author would like to thank the participants Russell Miller, Tobias G\"artner and Martin Ziegler, as well as the organizer Christine Ga\ss ner. Moreover, the second author would like to thank Anders Hansen for fruitful discussions on the solvability complexity index.

The work presented here benefited from the Royal Society International Exchange Grant IE111233. The second author is partially supported by the ERC inVest project.
\end{document}